%% file: main.tex
\documentclass[]{interact}

\usepackage{epstopdf}
\usepackage[caption=false]{subfig}
\usepackage{newtxtext,newtxmath}
\usepackage{times} 
\usepackage{amsmath} 
\usepackage{amssymb,mathtools}  

\usepackage{siunitx}
\usepackage{array}
\usepackage{stfloats}
\usepackage{algorithm}
\usepackage{algorithmic}
\usepackage{booktabs}
\usepackage[flushleft]{threeparttable}
\usepackage{bm}
\usepackage{enumerate}

\usepackage{cancel}
\usepackage[numbers,sort&compress]{natbib}

\bibpunct[, ]{[}{]}{,}{n}{,}{,}
\makeatletter
\def\NAT@def@citea{\def\@citea{\NAT@separator}}
\makeatother

\theoremstyle{plain}
\newtheorem{theorem}{Theorem}[section]
\newtheorem{lemma}[theorem]{Lemma}

\newtheorem{proposition}[theorem]{Proposition}
\newtheorem{condition}{Condition}{}

\theoremstyle{definition}
\newtheorem{definition}[theorem]{Definition}

\theoremstyle{remark}

\begin{document}

\articletype{ARTICLE TEMPLATE}

\title{An Explicit Discrete-Time Dynamic Vehicle Model with Assured Numerical Stability}

\author{
\name{
Guojian Zhan\textsuperscript{a}, Qiang Ge\textsuperscript{a},  
Haoyu Gao\textsuperscript{a}, Yuming Yin\textsuperscript{b}, Bin Zhao\textsuperscript{c} and Shengbo Eben Li\textsuperscript{a}
}  
\affil{\textsuperscript{a}School of Vehicle and Mobility, Tsinghua University, Beijing, China;\\
\textsuperscript{b}School of Mechanical Engineering, Zhejiang University of Technology, China;\\
\textsuperscript{c}School of Electrical and Electronic Engineering, Changchun University of Technology, China.}
}
\thanks{CONTACT Shengbo Eben Li. Email: lisb04@gmail.com}

\maketitle
\begin{abstract}
Numerical stability is of great significance for discrete-time dynamic vehicle model. Among the unstable factors, low-speed singularity stands out as one of the most challenging issues, which arises from that the denominator of tire side angle term only contains the vehicle longitudinal speed. Consequently, for the common low-speed and stop-start driving scenarios, the calculated tire slip angle will approach infinity, which will further lead to the numerical explosion of other vehicle states. In response to this critical challenge, we propose a discrete-time dynamic vehicle model that effectively mitigates the low-speed singularity issue, ensuring numerical stability and maintaining the explicit form—highly favored by model-based control algorithms. To validate the numerical stability of our model, we conduct a rigorous theoretical analysis, establishing sufficient conditions for stability, and conduct extensive empirical validation tests across a wide spectrum of speeds. Subsequent to the validation process, we conduct comprehensive simulations comparing our proposed model with both kinematic models and existing dynamic models discretized through the forward Euler method. The results demonstrate that our proposed model shows better comprehensive performance in terms of both the accuracy and numerical stability. Finally, the real vehicle experiments are carried out to support that our proposed model can closely aligns to the real vehicle trajectories showcasing its practicality and ease of use. Notably, our work stands as the pioneering endeavor in introducing an explicit discrete-time dynamic vehicle model suitable for common urban driving scenarios including low-speed and stop-start.
\end{abstract}

\begin{keywords}
Dynamic vehicle model; Numerical stability; Model-based control
\end{keywords}

\input{Section01}
\input{Section02}
\input{Section03}
\input{Section04}
\input{Section05}
\input{Section06}
\input{Section07}
\input{Section08}

\section*{Funding}
This work was supported by NSF China under 52221005, Tsinghua University Initiative Scientific Research Program, and Tsinghua-Toyota Joint Research Institute Inter-disciplinary Program.

\end{document}

%% file: Section01.tex
\section{Introduction}

In the ever-evolving landscape of transportation and mobility, the role of vehicle model has emerged as a foundational pillar in the fields of academia and industry \cite{sharp2011vehicle}. 
Understanding the behavior of vehicles within the discrete-time framework has taken on essential importance, driven by its direct relevance to controller design at general digital computers \cite{zhan2023continuous}. 
Up to this point, the modeling of a typical four-wheel ground vehicle with front steering has been approached at different scales, ranging from a basic 2-Degree-of-Freedom (DoF) lateral quarter vehicle to the full size vehicle \cite{yang2013overview}. In practical engineering, the single-track model takes the lead as the most commonly employed option. This model's brilliance lies in its strategic utilization of symmetry for simplification, and being able to distinguish motion characteristics of the front and rear axles \cite{kutluay2014validation}. That is, the single-track vehicle model strikes an effective equilibrium between computational efficiency and accuracy \cite{manrique2022analytical}. Therefore, in the subsequent content, the term ``vehicle model'' defaults to referring to the single-track vehicle model.

The vehicle models generally can be divided into two main categories: kinematic models and dynamic models \cite{kong2015kinematic}. Kinematic models concentrate exclusively on the geometric facets of motion, while dynamic models go further by encompassing factors like tire forces to achieve higher accuracy. Therefore, while the kinematic model suffices for straightforward tasks, the dynamic vehicle model takes center stage for precision-focused motion control needs.

However, for the dynamic vehicle model, the low-speed singularity has served as one of the most challenging factors of the numerical stability \cite{arnold2011numerical}. Although there have been many explicit discretization methods in mathematics, such as forward Euler method and Runge-Kutta method, none of these methods can adequately address the challenge posed by low-speed singularity \cite{biswas2013discussion}. The low-speed singularity  arises from that the denominator of tire slip angle term is solely dependent on longitudinal vehicle speed \cite{lugner2005recent}. This means that when the vehicle speed is extremely low and close to 0, it will cause the tire slip angle to approach infinity, which will result in a chain explosion of all vehicle state variables. 

Alongside the explicit discretization methods mentioned, there exists a category of implicit discretization methods, exemplified by the backward Euler method, which can provide numerical stability guarantee \cite{arnold2011numerical}. However, the resulting discrete-time vehicle model usually lacks an explicit form, necessitating iterative calculation for solving the next state at each step, significantly impairing the efficiency and failing to support the building of state transition relationships \cite{10056957}. In contrast, those model expressed explicitly as $x_{k+1}=f(x_k, u_k)$ but not implicitly as $x_{k+1}=f(x_{k}, x_{k+1}, u_k)$ are particularly favored, due to their high computational efficiency and ability to support the creation of state transition graphs \cite{li2023brlok}. Such an explicit model is crucial for model-based online optimization methods such as Model Predictive Control (MPC) \cite{li2010model} and model-based learning algorithms such as Approximate Dynamic Programming (ADP) 
 \cite{ren2023improve}. Here $x_k$ is the state, $u_k$ is the action, and the subscript $k$ denotes the time step. To summarize, none of the current dynamic vehicle model with explicit form are capable of addressing the issue of low-speed singularity, thereby lacking assurance of numerical stability.

In this paper, we propose a discrete-time dynamic vehicle model that can avoid low-speed singularity to ensure numerical stability while retaining the favored explicit form. Our key contributions can be summarized as:
\begin{itemize}
    \item (1) We propose a discrete-time dynamic vehicle model that can  achieve assured numerical stability while retaining the explicit form. Specifically, through using the backward Euler method for the lateral speed and yawrate and the forward Euler method for the other states, the denominator of tire slip angle term will no longer only contain the longitudinal vehicle speed, that is, the low-speed singularity can be avoided to enhance numerical stability. Besides, when combined with the common linear tire model, our proposed model can be expressed in the favored explicit form.
    \item (2) We provide a sufficient condition with detailed theoretical proof for ensuring the numerical stability. In addition, a bunch of numerical tests across a wide range of speeds are conducted to further validate it empirically. This assurance lays the foundation that our proposed model can be confidently embedded into model-based systems for general scenarios including low-speed and stop-start.
    \item (3) We conduct a series of simulations to claim that: compared with the kinematic model and the dynamic models with existing discretization methods, our proposed model shows a more comprehensive performance in terms of accuracy and numerical stability. In addition, our model is notably practical and user-friendly supported by an effective closed-loop control case at a typical stop-start driving scenario.
    \item (4) We further verify the effectiveness of our proposed model by real vehicle experiments. Through a series of open-loop control input of acceleration, deceleration and step angle, the distance error with the real vehicle trajectory is consistently kept within a reasonable range, which strongly evidents the practicality and the ease of use.
\end{itemize}

The subsequent sections are structured as follows: Section \ref{sec_pre} introduces the preliminaries, encompassing dynamic and kinematic single-track vehicle models along with classical discretization methods. Section \ref{sec_instability} delves into an exhaustive exposition of the instability concerns stemming from the low-speed singularity, along with the representative compromise adopted by \textit{CarSim} software in engineering. Our proposed explicit discrete-time dynamic model is outlined in Section \ref{sec_discrete}, accompanied by the details illustration of the discretization process. In Section \ref{sec_stability}, a comprehensive analysis of the model's numerical stability is presented, accompanied by a theoretical proof. Furthermore, the paper proceeds to conduct several simulations in Section \ref{sec_simulation} for the purpose of validating stability, accuracy, computational efficiency, and its suitability for stop-start driving condition. Section \ref{sec_experiment} underscores the practical utility of the proposed model by showcasing its good alignment with real vehicle trajectories. Finally, Section \ref{sec_conclusion} draws conclusions.

%% file: Section02.tex
\section{Preliminaries}
\label{sec_pre}

This section will firstly present the commonly employed continuous-time form of the dynamic and kinematic single-track vehicle model. Then it will introduce classical discretization methods for ordinary differential equation in mathematics, showcasing two representative instances including the forward and backward Euler methods.

\subsection{Dynamic model}

The dynamic model is given in (\ref{dyna_continuous}). The six states are respectively horizontal position $X$, vertical position $Y$, yaw angle $\varphi$, longitudinal velocity $U$, lateral velocity $V$ and yawrate $\omega$. The two control inputs are longitudinal acceleration $a$ and front wheel steering angle $\delta$. All the variables are illustrated in Figure \ref{dyna}.
\begin{subequations}
\begin{align}
    \dot{x}&=f(x, u)=
    \left[\begin{array}{c}
    U \cos \varphi-V \sin \varphi \\
    U \sin \varphi+V \cos \varphi\\
    \omega \\
    a+V \omega-\frac{1}{m} F_\text{f} \sin \delta \\
    -U \omega  + \frac{1}{m} \left (F_\text{f} \cos \delta+F_\text{r}\right)\\
    \frac{1}{I_\text{z}}\left(l_\text{f} F_\text{f} \cos \delta-l_\text{r} F_\text{r}\right)
    \end{array}\right], \\
    x&=\begin{bmatrix}
    X& Y& \varphi& U& V& \omega
    \end{bmatrix}^\top, \ u=\begin{bmatrix}
    a& \delta
    \end{bmatrix}^\top.
    \label{dyna_continuous} 
\end{align}
\end{subequations}

The lateral tyre forces $F_\text{f}$, $F_\text{r}$ are determined by sideslip angle $\alpha_\text{f}$ and $\alpha_\text{r}$. Since the extreme case that with large lateral acceleration is not under consideration, it is reasonable to assume
\begin{subequations}
    \begin{align}
    F_\text{f}&=k_\text{f} \alpha_\text{f} \approx k_\text{f}\left(\frac{V+l_\text{f} \omega}{U}-\delta\right)\label{c}, \\
    F_\text{r}&=k_\text{r} \alpha_\text{r} \approx k_\text{r} \frac{V-l_\text{r} \omega}{U}\label{d}.
\end{align}
\label{sideslipangle}
\end{subequations}

\subsection{Kinematic model}
The kinematic model  \cite{polack2017kinematic} given in (\ref{kine_continuous}) assumes that both front and rear wheels have only longitudinal rolling movement, but no lateral slip. This is less accurate because tire characteristics are neglected. The four state variables are respectively horizontal position $X$, vertical position $Y$, longitudinal velocity $U$ and yaw angle $\varphi$, which are labeled in Figure \ref{kina}.
\begin{subequations}
\begin{align}
    \dot{x}_{\text{kine}}&=f_{\text{kine}}\left(x_{\text{kine}}, u\right)=\left[\begin{array}{c}
    U \cos \varphi-\frac{l_\text{r}}{l_\text{f}+l_\text{r}} U \tan \delta \sin \varphi \\
    U \sin \varphi+\frac{l_\text{r}}{l_\text{f}+l_\text{r}} U \tan \delta  \cos \varphi \\
    a \\
    \frac{1}{l_\text{f}+l_\text{r}} U \tan \delta \\
    \end{array}\right], \\
    x_\text{kine}&=\begin{bmatrix}
    X& Y& \varphi& U
    \end{bmatrix}^\top, \  
    u=\begin{bmatrix}
    a& \delta
    \end{bmatrix}^\top.
    \label{kine_continuous}
\end{align}
\end{subequations}

\begin{figure}[htbp]
    \centering
    \captionsetup[subfigure]{justification=centering}
        \subfloat[Dynamic]{\label{dyna}\includegraphics[width=0.35\textwidth]{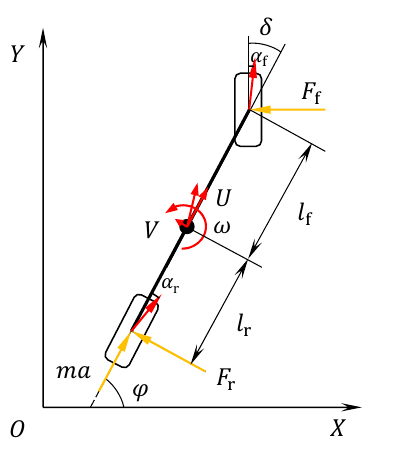}} \quad
        \subfloat[Kinematic]{\label{kina}\includegraphics[width=0.35\textwidth]{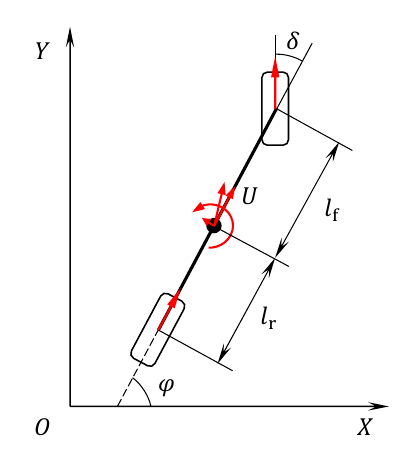}}
    \caption{Vehicle model}
\label{bicycle models}
\end{figure}

\subsection{Classical discretization methods}
In order to apply the ordinary differential equation model to design controller at digital computers, it is necessary to implement numerical discretization. As typical representatives of the explicit method and the implicit method, forward Euler method and backward Euler method are presented as shown in \eqref{forward_euler} and \eqref{back_euler}, respectively. Where footnote $k$ is the time step and $T_{s}$ is the discrete time step.
\begin{subequations}
\begin{align}
    x_{k+1}=T_\text{s} f\left(x_{k}, u_{k}\right)+x_{k},
    \label{forward_euler}\\
    x_{k+1}-T_\text{s} f\left(x_{k+1}, u_{k}\right)=x_{k}.
    \label{back_euler}  
\end{align}
\end{subequations}

As a representative implicit technique, the numerical stability of the backward Euler method is often superior. However, in implicit methods, $x_{k+1}$ needs to be determined through iterative solution of equation \eqref{back_euler}, where fixed-point iteration is typically employed. Thus, if the forward Euler method is directly applied, the discretized vehicle dynamics model may not meet the numerical stability requirements encountering low-speed driving conditions. Conversely, if the backward Euler method is directly employed, the discrete model lacks explicit form and is generally unsuitable for model-based planning or learning methods, such as MPC and ADP, to implement rollout \cite{biswas2013discussion}.

%% file: Section03.tex
\section{Instability Challenge and Existing Compromise}
\label{sec_instability}
This section will first detail the low-speed singularity and how it leads to numerical instability, and then introduce the existing compromise adopted by the professional industrial software \textit{CarSim} to workaround such a challenge.

\subsection{Low-speed singularity}
The Low-speed singularity is a longstanding challenge for dynamic model. This issue arises from the fact that the tire slip angle's estimated term in equation (\ref{sideslipangle}) relies solely on the longitudinal vehicle speed. It's important to highlight that regardless of the chosen tire force model, be it Magic Formula \cite{pacejka1992magic}, the Fiala model \cite{lugner2005recent}, or the TMeasy model \cite{hirschberg2007tire}, the estimation of $\alpha_\text{f}$ and $\alpha_\text{r}$ using (\ref{sideslipangle}) remains. Consequently, the low-speed singularity introduced by the longitudinal speed $U$ in the denominator becomes an unavoidable factor.

In urban scenarios, stop-start is an common situation, where a vehicle will gradually slow down due to the red traffic light or stationary front vehicle and then accelerates again. As the vehicle's longitudinal velocity $U$ approaches zero, the estimations of the tire slip angles $\alpha_\text{f}$ and $\alpha_\text{r}$ may inaccurately tend towards zero, potentially leading to imprecise evaluations or even an infinite value of the lateral tire forces $F_\text{f}$ and $F_\text{r}$.  This phenomenon will ultimately result in the explosion of the vehicle's lateral states ($V$ and $\omega$).

\subsection{Existing compromise}
In the field of vehicle dynamics industrial software, \textit{CarSim} stands as one of the most professional representatives.
When precise parameters, sufficiently small time precision and cumbersome empirical calibration data are employed, the software can reliably simulate intricate nonlinear vehicle systems. 
To workaround the low-speed singularity issue, the compromise adopted by \textit{CarSim} \cite{bernard1995tire} includes:

\begin{itemize}
    \item \textbf{Step 1: Hysteresis} Define the hysteresis slip angle as a new state variable, replacing the slip angle as the input of the tire model. The hysteresis slip angle of a tire is defined as: $\alpha_{\text{L,f}}=\arctan\left(\kappa_{\text{f}}\right), \alpha_{\text{L,r}}=\arctan\left(\kappa_{\text{r}}\right).$, where $\kappa_{\text{f}}$ and $\kappa_{\text{r}}$ are the hysteresis state variables of the front and rear tires respectively. Their dynamics satisfies
    \begin{equation}
        \dfrac{\mathrm{d}\kappa_{\text{f}}}{\mathrm{d}t}=\dfrac{{V}_{\text{f}}-\kappa_{\text{f} }\left|{U}_{\text{f}}\right|}{L_{\text{y}}},
        \dfrac{\mathrm{d}\kappa_{\text{r}}}{\mathrm{d}t}=\dfrac{{V}_{\text{r}}-\kappa_{\text{r} }\left|{U}_{\text{r}}\right|}{L_{\text{y}}}.
    \end{equation}
    
    ${U}_{\text{f}},{V}_{\text{f}}$ are the tangential and normal velocities of the front wheel in the center plane of the tire; ${U}_{\text{r}},{V}_{\text{r}}$ are the counterparts of the rear wheel. $L_{\text{y}}$ is the relaxation length constant. This compromise suppresses high-frequency oscillations through low-pass filtering, and increases the low-speed feasible range.

    \item \textbf{Step 2: Lower limit bound} Introduce a low speed limit, when the longitudinal speed is lower than this value, the denominator term in the formula takes a fixed value. Introduce $U_{\text{low},\alpha}$ as speed lower bound :
    \begin{equation}
            \left|{U}_{\text{f}} \right| = \max\left(\left|U_{\text{f}}\right|, {U}_{\text{low},\alpha} \right), 
            \left|{U}_{\text{r}} \right| = \max\left(\left|U_{\text{r}}\right|, {U}_{\text{low},\alpha} \right).
    \end{equation}
\end{itemize}

In summary, \textit{CarSim} handles the low-speed singularity problem through introducing new cumbersome state variable and modeling in pieces. 
As mentioned, such a implicit form compromise may have a good effect in offline simulation, but it does not meet the requirements of model-based algorithms.

%% file: Section04.tex
\section{Method}\label{sec_discrete}

This section will firstly provide an in-depth explanation of the derivation process of our proposed dynamic vehicle model, highlighting its emphasis on enhancing numerical stability. Subsequently, it will present its complete explicit expression.

\subsection{Highly stable explicit discretization}

This subsection dives into dealing with the low-speed singularity issue by designing novel discretization process. Our key idea is that taking advantage of the explicit form of the forward Euler method and the numerical stability of backward Euler method, to design an explicit dynamic vehicle model with assured numerical stability.

Recall that the main defect of backward Euler method, as a typical implicit method, is that obtaining $x_{k+1}$ involves solving a root-finding problem (\ref{back_euler}) at each step, usually using fixed-point iteration. In order to avoid solving such a root-finding problem and retain the explicit form, we derive the next state $x_{k+1}$ in a variable-by-variable manner. Specifically, 
$X_{k+1}$, $Y_{k+1}$ and $\varphi_{k+1}$ are calculated by forward Euler method since their  derivatives do not contain themselves. In other words, there is no risk of numerical divergence for these three variables.
$U_{k+1}$ is also derived by forward Euler method after simplifying $v \omega-\frac{1}{m} F_\text{f} \sin \delta$ as zero. This can be regarded as neglecting the longitudinal resistance force ($-\frac{1}{m} F_\text{f} \sin \delta \rightarrow 0$) first, and then substituting local (vehicle-attach) coordinate for ground coordinate system, i.e., take $a \rightarrow \dot{U}$ instead of $a + V \omega \rightarrow \dot{U}$.

After that, $V_{k+1}$ and $\omega_{k+1}$ are solved in an implicit manner as:
\begin{subequations}
    \begin{align}
    V_{k+1} &= V_k + T_\text{s}f_5\left([X_k, Y_k, \varphi_k, U_k, V_{k+1}, \omega_k]^\top, u_k\right), \\
    \omega_{k+1} &= \omega_k + T_\text{s}f_6\left([X_k, Y_k, \varphi_k, U_k, V_k, \omega_{k+1}]^\top, u_k\right),
    \end{align}\\
    \label{eq: back_last_two_variable}
\end{subequations}
where the subscript $i$ of $f_i$ denotes its $i$-th component. Note that all $k+1$-step variables but the unsolved $V_{k+1}$, $\omega_{k+1}$ are approximated by $k$-step states respectively, compared to the standard backward Euler formulation. Subtly the above equations \eqref{eq: back_last_two_variable} have explicit expressions when employing linear tyre model, which can accommodate most common driving situations. 

Technically the discretization is neither a standard forward nor backward Euler method, but a variant inspired by them. It is the adoption of backward conception that makes the model stable, which will be theoretically proved in Section \ref{sec_stability}.

\subsection{Highly stable explicit dynamic vehicle model}

For the first four state variables $[X, Y, \varphi, U]^\top$ using forward Euler method, their updating expressions are
\begin{equation}
    \left[\begin{array}{c}
    X_{k+1} \\
    Y_{k+1} \\
    \varphi_{k+1} \\
    U_{k+1} \\
    \end{array}\right]=\left[\begin{array}{c}
    X_{k}+T_\text{s}\left(U_{k} \cos \varphi_{k}-V_{k} \sin \varphi_{k}\right) \\
    Y_{k}+T_\text{s}\left(V_{k} \cos \varphi_{k}+U_{k} \sin \varphi_{k}\right) \\
    \varphi_{k}+T_\text{s} \omega_{k} \\
    U_{k}+T_\text{s} a_{k}
    \end{array}\right].
    \label{first_four_variable}
\end{equation}

As for the last two state variables $[V, \omega]^\top$ using a variant version of backward Euler method, we first calculate the tyre side slip forces as \begin{subequations}
    \begin{align}
    \tilde{F}_{\text {f}}& \coloneqq F_{\text {f}}\left({U}_{k}, {V}_{k+1}, \omega_k, \delta_k\right)=k_{\mathrm{f}}\left(\frac{{V}_{k+1}+l_{\mathrm{f}} \omega_k}{{U}_{k}}-\delta_k\right),\\
    \tilde{F}_{\text {r}}& \coloneqq F_{\text {r}}\left({U}_{k}, {V}_{k+1}, \omega_k\right)=k_{\mathrm{r}}\left(\frac{{V}_{k+1}-l_{\mathrm{r}} \omega_k}{{U}_{k}}\right), \\
    \hat{F}_{\text {f}}& \coloneqq F_{\text {f}}\left({U}_{k}, {V}_{k}, \omega_{k+1}, \delta_k\right)=k_{\mathrm{f}}\left(\frac{{V}_{k}+l_{\mathrm{f}} \omega_{k+1}}{{U}_{k}}-\delta_k\right),\\
    \hat{F}_{\text {r}}& \coloneqq F_{\text {r}}\left({U}_{k}, {V}_{k}, \omega_{k+1}\right)=k_{\mathrm{r}}\left(\frac{{V}_{k}-l_{\mathrm{r}} \omega_{k+1}}{{U}_{k}}\right).
    \end{align}
    \label{eq.tyre_forces}
\end{subequations}
Then following \eqref{eq: back_last_two_variable}, we can express the updating of $[V, \omega]^\top$ as
    \begin{subequations}
        \begin{align}
        V_{k+1}=V_{k} + T_\text{s}\left(-U_k \omega_k+\frac{1}{m}\left(\tilde{F}_\text{f}(U_k,V_{k+1},\omega_k,\delta_{k})+\tilde{F}_\text{r}(U_k,V_{k+1},\omega_k)\right)\right), \\
        \omega_{k+1}=\omega_{k} +T_\text{s} \left( \frac{1}{I_\text{z}}\left(l_\text{f} \hat{F}_\text{f}(U_k,V_k,\omega_{k+1},\delta_{k}) -l_\text{r} \hat{F}_\text{r}(U_k,V_k,\omega_{k+1})\right)\right).
        \end{align}\\
        \label{eq: back_V_and_omega_expand}
    \end{subequations}
By substituting \eqref{eq.tyre_forces} into \eqref{eq: back_V_and_omega_expand}, we can acquire the expanded form as
\begin{equation}
    \left[\begin{array}{c}
    V_{k+1} \\
    \omega_{k+1}
    \end{array}\right]=\left[\begin{array}{c}
    \dfrac{m U_{k} V_{k}+T_\text{s}\left(l_\text{f}k_\text{f}-l_\text{r} k_\text{r}\right) \omega_{k}-T_\text{s} k_\text{f} \delta_{k} U_{k}-T_\text{s} m U_{k}^{2} \omega_{k}}{m U_{k}-T_\text{s}\left(k_\text{f}+k_\text{r}\right)} \\
    \dfrac{I_\text{z} U_{k} \omega_{k}+T_\text{s}\left(l_\text{f} k_\text{f}-l_\text{r} k_\text{r}\right) V_{k}-T_\text{s} l_\text{f} k_\text{f} \delta_{k} U_{k}}{I_\text{z} U_{k}-T_\text{s}\left(l_\text{f}^{2} k_\text{f}+l_\text{r}^{2} k_\text{r}\right)}
    \end{array}\right].
    \label{last_two_variable}
\end{equation}

Finally, by combining \eqref{first_four_variable} and \eqref{last_two_variable}, the complete form of our proposed explicit dynamic vehicle model is
\begin{equation}
    \left[\begin{array}{c}
    X_{k+1} \\
    Y_{k+1} \\
    \varphi_{k+1} \\
    U_{k+1} \\
    V_{k+1} \\
    \omega_{k+1}
    \end{array}\right]=\left[\begin{array}{c}
    X_{k}+T_\text{s}\left(U_{k} \cos \varphi_{k}-V_{k} \sin \varphi_{k}\right) \\
    Y_{k}+T_\text{s}\left(V_{k} \cos \varphi_{k}+U_{k} \sin \varphi_{k}\right) \\
    \varphi_{k}+T_\text{s} \omega_{k} \\
    U_{k}+T_\text{s} a_{k} \\
    \dfrac{m U_{k} V_{k}+T_\text{s}\left(l_\text{f}k_\text{f}-l_\text{r} k_\text{r}\right) \omega_{k}-T_\text{s} k_\text{f} \delta_{k} U_{k}-T_\text{s} m U_{k}^{2} \omega_{k}}{m U_{k}-T_\text{s}\left(k_\text{f}+k_\text{r}\right)} \\
    \dfrac{I_\text{z} U_{k} \omega_{k}+T_\text{s}\left(l_\text{f} k_\text{f}-l_\text{r} k_\text{r}\right) V_{k}-T_\text{s} l_\text{f} k_\text{f} \delta_{k} U_{k}}{I_\text{z} U_{k}-T_\text{s}\left(l_\text{f}^{2} k_\text{f}+l_\text{r}^{2} k_\text{r}\right)}
    \end{array}\right].
    \label{back_euler_bicycle}
\end{equation}

%% file: Section05.tex
\section{Numerical Stability Analysis}\label{sec_stability}
This section will firstly introduce the definition of numerical stability, followed by an analysis from error propagation perspective. Ultimately, it will provide a theoretical proof of the sufficient  numerical stable condition of our proposed dynamic model.
\subsection{Definition of numerical stability}

For the dynamic model in \eqref{back_euler_bicycle}, The first three state variables $[{X}_{k}$, ${Y}_{k}$, $\varphi_k]^\top$ are essentially the integral of the other three variables $[{U}_{k}$, ${V} _{k}$, $\omega_k]^\top$, so the intrinsic characteristics of the system can be fully described by only three variables, i.e., ${U}_{k}$, ${V}_{k}$ and $\omega_k$, since there are only three degrees of freedom in total. For the sake of simplicity, the three-variable dynamical system \eqref{simplified_back_euler_bicycle} will be used to replace \eqref{back_euler_bicycle}, and the notion $f(\cdot,\cdot)$ will be retained.

\begin{equation}
    \left[\begin{array}{c}
    U_{k+1} \\
    V_{k+1} \\
    \omega_{k+1}
    \end{array}\right]=\left[\begin{array}{c}
    U_{k}+T_\text{s} a_{k} \\
    \dfrac{m U_{k} V_{k}+T_\text{s}\left(l_\text{f}k_\text{f}-l_\text{r} k_\text{r}\right) \omega_{k}-T_\text{s} k_\text{f} \delta_{k} U_{k}-T_\text{s} m U_{k}^{2} \omega_{k}}{m U_{k}-T_\text{s}\left(k_\text{f}+k_\text{r}\right)} \\
    \dfrac{I_\text{z} U_{k} \omega_{k}+T_\text{s}\left(l_\text{f} k_\text{f}-l_\text{r} k_\text{r}\right) V_{k}-T_\text{s} l_\text{f} k_\text{f} \delta_{k} U_{k}}{I_\text{z} U_{k}-T_\text{s}\left(l_\text{f}^{2} k_\text{f}+l_\text{r}^{2} k_\text{r}\right)}
    \end{array}\right].
    \label{simplified_back_euler_bicycle}
\end{equation}

\begin{definition}
(Numerical stability of vehicle model).
For the system described by \eqref{simplified_back_euler_bicycle}, if for any $x_0 \in [x_{\min}, x_{\max}]$, $y_0 \in [x_{\min}, x_{\max}]$ and given control sequence $u_0, u_1, \cdots, u_{k-1}$, there exists a constant $C \in \mathbb{R}$ such that $\lim_{k \rightarrow \infty}\left\|x_k-y_k\right\| \leq C$, then the discrete dynamic system $x _{k+1}=f\left({x}_k, u_k\right)$ can be called numerically stable. Here, the $x_k$ and $y_k$ are the $k$th state, staring from the initial states $x_0, y_0$ respectively under the identical control sequence $u_0, u_1, \cdots, u_{k-1}$.
\label{definition 2-1}
\end{definition}

As shown in the Figure \ref{stab_illus}, a physical interpretation of the numerical stability is given. The plane represents the state space, and $\varepsilon_0=x_0-y_0$ is the error vector. It can be seen that this error vector moves in the state space as the control sequence $u_0, u_1, \cdots, u_{k-1}$ executing. The definition \ref{definition 2-1} requires that, with $k\rightarrow \infty$, $\varepsilon_\infty$ always remains within a sphere of bounded radius $C$. In fact, any different starting point in the domain can be regarded as a disturbance relative to a certain starting point; and the trajectory point at any time can be regarded as the starting point of the subsequent trajectory. In other words, this definition can be understood as requiring bounded nature of the error propagation under any potential disturbances.

\begin{figure}[ht]
  \centering
  \includegraphics[width=0.25\linewidth]{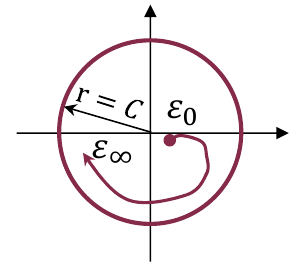}
  \caption{The physical meaning of numerical stability definition.}
  \label{stab_illus}
\end{figure}

\subsection{Analysis from error propagation perspective}
In order to better analyze the numerical stability property from error propagation perspective, this subsection expresses the recursive process of the trajectory equivalently as a series of multiplicative forms of different linear matrices. 
Specifically, The nonlinear recurrence equation \eqref{simplified_back_euler_bicycle} can be viewed as a reference point $\left(x_k, u_k, x_{k+1}\right)$ on the following nonlinear continuous function $y = f(x,u)$, where $x$, ${u}$ is the independent variable, and $y$ is the dependent variable.

\begin{lemma}[Lagrange Mean Value Theorem]
     \label{lemma1}
     Let $g:[c, d] \rightarrow \mathbb{R}$ be a continuous function on the closed interval $\left[c,d\right]$, and in the open interval $\left(c,d\right)$ upperly differentiable, $c<d$. On $\left(c,d\right)$, there exists a number $m$ satisfying
     \begin{equation}
         g^{\prime}(m)=\frac{f(d)-f(c)}{d-c}.
         \nonumber
     \end{equation}
\end{lemma}

Since $u_k$ is often known and bounded, it can be regarded as a time-varying parameter. 
Consequently, $x_{k+1}$ becomes a function solely reliant on the independent variable $x_k$. As indicated by Lemma \ref{lemma1}, within this continuously differentiable function, it is always possible to identify ${\tau}_k=\left[{\tau}_k^{(1)}, {\tau}_k^{(2)}, {\tau}_k^{(3)}\right]$, where ${\tau}_k^{(1)}, {\tau}_k^{(2)}, {\tau}_k^{(3)} \in[0,1]$, and this satisfies
\begin{subequations}
    \begin{align}
        y-x_{k+1} &={A}_k\left(\tau_k\right)\left(x-x_k \right), \label{eq 2-22} \\
        {A}_k\left(\tau_k\right)&=\left[\begin{array}{l}
     \left.\frac{\partial f^{(1)}}{\partial x^{\top}}\right|_{x=x_k+\tau_k^{(1)}\left(x-x_k\right)} \\
     \left.\frac{\partial f^{(2)}}{\partial x^{\top}}\right|_{x=x_k+\tau_k^{(2)}\left(x-x_k\right)} \\
     \left.\frac{\partial f^{(3)}}{\partial x^{\top}}\right|_{x=x_k+\tau_k^{(3)}\left(x-x_k\right)} 
\end{array}\right] \label{eq 2-23}
    \end{align}
\end{subequations}

where $\left.\frac{\partial f^{(j)}}{\partial x^{\top}}\right|_{x=x_k+{ {\tau}}_k^{(j)}\left(x-x_k\right)}$ is the Jacobian matrix of the $j$-th component of the function $f\left(x, u =u_k\right)$ at $x_k+\tau_k^{(j)}\left(x-x_k\right)$, $j=1,2,3$. It should be pointed out that since the Lagrange Mean Value Theorem of real
functions does not hold for vector functions, it is not guaranteed to find the vector $\hat{\tau} \in[0,1]$, satisfying
$
     y-x_{k+1}=\left.\frac{\partial f}{\partial x^{\top}}\right|_{x=x_k+\hat{\tau}\left(x-x_k\right)}\left(x-x_k\right).
     \nonumber
$
Therefore, it is necessary to introduce the vector $\tau_k$ and use the Lagrange Mean Value Theorem three times individually to obtain the equation \eqref{eq 2-22}. Let's denote $(x-x_k)$ as $\varepsilon_k$. Notice that in the process of solving the numerical solution of the differential equation, the repeated call of the equation \eqref{simplified_back_euler_bicycle} is actually equivalent to the dependent variable $y$ being continuously input into the function $f(x,u)$. Therefore, $(y-x_{k+1})$ can be denoted as $\varepsilon_{k+1}$, the equation \eqref{eq 2-22} becomes
\begin{equation}
    \varepsilon_{k+1}={A}_k\left(\tau_k\right) \varepsilon_k.
    \label{eq 2-24}
\end{equation}

The equation \eqref{eq 2-24} describes the disturbance on state it will lead to the next time step, following the system dynamics \eqref{simplified_back_euler_bicycle}. And the amplification relationship between disturbances is determined by ${A}_k\left(\tau_k\right)$. According to the definition of the Jacobian matrix, we have
\begin{equation}
    {A}_k\left(\tau_k\right)
      \coloneqq\left[\begin{array}{cc}
     1 &\textbf{0} \\
     b_k\left(\tau_k\right) & \hat{{A}}_k\left(\tau_k\right)
     \end{array}\right].
     \label{eq 2-25}
\end{equation}
Then further expand to get:
\begin{equation}
     {A}_k\left(\tau_k\right)=\left[\begin{array}{ccc}
     1 & 0 & 0 \\
     b^{(1)}_{k}\left(\tau_k^{(2)}\right) & \dfrac{m {U}_{k}\left({{\tau}}_k^{(2)}\right)}{m {U}_{k}\left(\tau_k^{(2)}\right)-T_\text{s}\left(k_{\mathrm{f}}+k_\text{r}\right)} & \dfrac{T_\text{s}\left(l_\text{f} k_\text{f}-l_{\mathrm{r}} k_\text{r}\right)-T_\text{s} m {U}_{k}^2\left(\tau_k^{(2)}\right)}{m {U}_{k}\left(\tau_k^{(2)}\right)-T_\text{s}\left(k_\text{f}+k_\text{r}\right)} \\
     b^{(2)}_{k}\left(\tau_k^{(3)}\right) & \dfrac{T_\text{s}\left(l_\text{f} k_\text{f}-l_\text{r} k_\text{r}\right)}{I_\text{z} {U}_{k}\left({{\tau}}_k^{(3)}\right)-T_\text{s}\left(l_\text{f}^2 k_\text{f}+l_\text{r}^2 k_\text{r}\right)} & \dfrac{I_\text{z} {U}_{k}\left(\tau_k^{(3)}\right)}{I_\text{z} {U}_{k}\left(\tau_k^{(3)}\right)-T_\text{s}\left(l_\text{f}^2 k_{\mathrm{f}}+l_\text{r}^2 k_\text{r}\right)}
     \end{array}\right], \\
\label{eq 2-25 expand}
\end{equation}
where
\begin{equation}
\begin{aligned}
    b^{(1)}_{k}\left(\tau_k^{(2)}\right)\!=&\frac{1}{\left[m {U}_{k}\left(\tau_k^{(2)}\right)-T_\text{s}\left(k_\text{f}+k_\text{r}\right)\right]^2}\Bigg\{\left[m {V}_{k}\left(\tau_k^{(2)}\right)\right. \\
    & -\left.\left(T_\text{s} k_\text{f} \delta_k+2 T_\text{s} m {U}_{k}\left(\tau_k^{(2)}\right) \omega_k\right)\right]\left[m {U}_{k}\left(\tau_k^{(2)}\right)-T_\text{s}\left(k_\text{f}+k_\text{r}\right)\right]\\
    &-m\left[m {U}_{k}\left(\tau_k^{(2)}\right) {V}_{k}\left(\tau_k^{(2)}\right)\right. -T_\text{s} k_\text{f} \delta_k {U}_{k}\left(\tau_k^{(2)}\right) \\
    &+T_\text{s}\left(l_\text{f} k_\text{f}-l_\text{r} k_\text{r}\right) \omega_k\left(\tau_k^{(2)}\right) -T_\text{s} m {U}_{k}^2\left(\tau_k^{(2)}\right) \left.\omega_k\left(\tau_k^{(2)}\right)\right] \!\Bigg\}, \\
    \end{aligned}
\label{eq 2-26}
\end{equation}
\begin{equation}
\begin{aligned}
    b^{(2)}_{k}\left(\tau_k^{(3)}\right)=&\frac{1}{\left[I_\text{z} {U}_{k}\left(\tau_k^{(3)}\right)-T_\text{s}\left(l_\text{f}^2 k_\text{f}+l_\text{r}^2 k_\text{r}\right)\right]^2}\Bigg\{ \left[I_\text{z} {U}_{k}\left(\tau_k^{(3)}\right)\right. \\
    &-T_\text{s}\left(l_\text{f}^2 k_\text{f}+l_\text{r}^2 k_\text{r}\right)\Big]\left(I_\text{z} \omega_k\left(\tau_k^{(3)}\right)-T_\text{s} l_\text{f} k_\text{f} \delta_k\right)\\
    & -I_\text{z}\left[I_\text{z} {U}_{k}\left(\tau_k^{(3)}\right) \omega_k\left(\tau_k^{(3)}\right)+T_\text{s}\left(l_\text{f} k_\text{f}-l_\text{r} k_\text{r}\right) {V}_{k}\left(\tau_k^{(3)}\right)\right. \\
    &-\left.T_\text{s} l_\text{f} k_\text{f} \delta_k {U}_{k}\left(\tau_k^{(3)}\right)\right]\Bigg\}.
\end{aligned}
\label{eq 2-27}
\end{equation}

The three involved $\tau_k$ variables are defined as
\begin{equation}
\begin{aligned}
     {U}_k\left(\tau_k^{(j)}\right) & =U_k+\tau_k^{(j)}\left({U}-{U}_k\right), \\
     {V}_{k}\left(\tau_k^{(j)}\right) & ={V}_{k}+\tau_k^{(j)}\left({V}-{V}_{k}\right), \\
     \omega_k\left(\tau_k^{(j)}\right) & =\omega_k+\tau_k^{(j)}\left(\omega-\omega_k\right).
\end{aligned}
\nonumber
\end{equation}

According to equation \eqref{eq 2-25}, it is easy to deduce
\begin{equation}
\begin{aligned}
     {A}_k\left(\tau_k\right) \cdots {A}_1\left(\tau_1\right) {A}_0\left(\tau_0 \right)&=\left[\begin{array}{cc}
     1 & \textbf{0} \\
     b_k\left(\tau_k\right) & \hat{{A}}_k\left(\tau_k\right)
     \end{array}\right] \cdots\left[\begin{array}{cc}
     1 & \textbf{0} \\
     b_0\left(\tau_0\right) & \hat{{A}}_0\left(\tau_0\right)
     \end{array}\right] \\
     & \coloneqq\left[\begin{array}{cc}
     1 & \textbf{0} \\
     b_{k\sim 0}\left(\tau_k, \cdots, \tau_0\right) & \hat{{A}}_{k\sim 0 }\left(\tau_k, \cdots, \tau_0\right)
     \end{array}\right].
\end{aligned}
\label{eq 2-28}
\end{equation}

\subsection{Numerical stability analysis}
Numerical stability defined in Definition \eqref{definition 2-1} is an important property of a model. Note that the discrete dynamic vehicle model shown in \eqref{back_euler_bicycle} is derived using our proposed novel discretization method, and the longitudinal vehicle speed is no longer the only term in any denominators. This means that the model is no longer singular at low speeds. This subsection further gives a theoretical analysis of the numerical stability.

We first present the Condition \ref{condition 1}, which serves as a sufficient condition for numerical stability. Then we raise the Proposition \eqref{proposition 2-1}, and provide a detailed proof as follows.

\begin{condition}
     The block matrix $\hat{{A}}_k\left(\tau_k\right)$ in \eqref{eq 2-25}, for $\forall k\in \mathbb{N} $ all satisfy $\left\|\hat{{A}}_k\left(\tau_k\right)\right\| \leq 1.$
     \label{condition 1}
\end{condition}
\begin{proposition}
If the Condition \ref{condition 1} holds, then the system described by \eqref{simplified_back_euler_bicycle} fully satisfies the numerical stability under the Definition \ref{definition 2-1}.
\label{proposition 2-1}
\end{proposition}
\begin{proof}
     \begin{equation}
     \begin{aligned}
     \left\|x_k-y_k\right\| & =\left\|\varepsilon_k\right\| \\
     & =\left\|{A}_{k-1}\left(\tau_{k-1}\right) \varepsilon_{k-1}\right\| \\
     & =\left \|{A}_{k-1}\left(\tau_{k-1}\right) \cdots {A}_1\left(\tau_1\right) {A}_0\left(\tau_0\right) \varepsilon_0 \right\|.
    \end{aligned}
    \end{equation}

Since $\varepsilon_0$ is the difference between two state points in the closed domain, the geometric meaning of the norm of this term is the distance between two points in the bounded space, so this item $\varepsilon_0$ must be bounded. 

We then proceed to establish the boundedness of the norm of the coefficient matrix ${A}_{k-1}\left(\tau_{k-1}\right) \cdots {A}_1\left(\tau_1\right) {A}_0\left(\tau_0\right)$. By referring to \eqref{eq 2-28}, it can be proved that both the block matrix $\hat{{A}}_{k-1\sim 0}\left({{\tau}}_{k-1}, \cdots, {{\tau}}_0\right)$ and the block vector $b_{k-1\sim 0}\left(\tau_{k-1}, \cdots, \tau_0\right)$ exhibit bounded norms. This substantiates the conclusion that the entire coefficient matrix ${A}_{k-1}\left(\tau_{k-1}\right) \cdots {A}_1\left(\tau_1\right) {A}_0\left(\tau_0\right)$ is indeed bounded. In the following, we will prove their bounded properties respectively.

\begin{itemize}
    \item  \textbf{$\left\|\hat{{A}}_{k-1\sim 0}\left(\tau_{k-1}, \cdots, \tau_0\right)\right\|$ is bounded.}

    \begin{equation}
     \begin{aligned}
     \left\|\hat{{A}}_{k-1\sim0}\left(\tau_{k-1}, \cdots, \tau_0\right)\right \| & =\left\|\hat{{A}}_{k-1}\left(\tau_{k-1}\right) \cdots \hat{{A}}_1\left(\tau_1\right) \hat{{A}}_0\left(\tau_0\right)\right\| \\
     & \leq\left\|\hat{{A}}_{k-1}\left(\tau_{k-1}\right)\right\| \cdots\left\|\hat {{A}}_1\left(\tau_1\right)\right\|\left\|\hat{{A}}_0\left(\tau_0\right)\right \| \\
     & \leq 1.
     \end{aligned}
     \label{eq 2-30}
\end{equation}

    \item \textbf{$\left\|b_{k-1\sim 0}\left(\tau_{k-1}, \cdots, \tau_0\right) \right\|$ is bounded.}

    \begin{equation}
     \begin{aligned}
     b_{k-1\sim0}\left(\tau_{k-1}, \cdots, \tau_0\right) & =b_ {k-1}\left(\tau_{k-1}\right)+\hat{{A}}_{k-1}\left(\tau_{k- 1}\right) b_{k-2\sim0}\left(\tau_{k-2}, \cdots, \tau_0\right) \\
     & =\sum_{j=0}^{k-1} \prod_{i=j+1}^{k-1} \hat{{A}}_i\left(\tau_i\right ) b_j\left(\tau_j\right).
     \end{aligned}
     \label{eq 2-31}
\end{equation}

\begin{equation}
     \begin{aligned}
     \left\|b_{k-1\sim0}\left(\tau_{k-1}, \cdots, \tau_0\right)\right\| & =\left\|\sum_{j=0}^{k-1} \prod_{i=j+1}^{k-1} \hat{{A}}_i\left({{\tau} }_i\right) b_j\left(\tau_j\right)\right\| \\
     & \leq \sum_{j=0}^{k-1} \prod_{i=j+1}^{k-1}\left\|\hat{{A}}_i\left({{\tau }}_i\right)\right\|\left\|b_j\left(\tau_j\right)\right\| \\
     & \leq \sum_{j=0}^{k-1}\left\|\hat{{A}}_*\right\|^{k-1-j}\left\|b _*\right\| \\
     & =\left\|b_*\right\| \frac{1-\left\|\hat{{A}}_*\right\|^k}{1-\left\|\hat{{A}}_*\right\|},
     \end{aligned}
     \label{eq 2-32}
\end{equation}

where
\begin{subequations}
    \begin{align}
       &\left\|b_*\right\|=\max _{x_j, y_j, \tau_j, T_\text{s}}\left\|{{ b}}_j\left(\tau_j\right)\right\|, 0 \leq j \leq k-1,
        \label{eq 2-33} \\
        &\left\|\hat{{A}}_*\right\|=\max _{x_i, y_i, \tau_i, T_\text{s}}\left\| \hat{{A}}_i\left(\tau_i\right)\right\|, j+1 \leq i \leq k-1.
\label{eq 2-34}
    \end{align}
\end{subequations}

In the derivation presented in \eqref{eq 2-32}, a sequence of mathematical tools is employed, including the Triangle Inequality and the Submultiplication of Matrix Norm. Note that $b_j\left(\tau_j\right)$ is defined in the closed domain that $x_{\text {min}}\leq x_j ,y_j\leq x_{\text {max}}$, $0\leq \tau^{(p)}_j\leq 1, p=1,2, 3$, $T_\text{s}>0$; $\hat{{A}}_i\left(\tau_i\right)$ is defined in the closed domain that $x_{\text {min}}\leq x_i,y_i\leq x_{\text {max}}$, $0\leq \tau^{(p )}_i\leq 1, p=1,2,3$, $T_\text{s}>0$. According to the Extreme value theorems of real functions, the maximum in \eqref{eq 2-33} and \eqref{eq 2-34} must exist.

The final step of the derivation in \eqref{eq 2-32} has not yet addressed the scenario where $\left\|\hat{{A}}_*\right\|$ equals exactly 1. When $\left\|\hat{{A}}_*\right\|=1$, the situation becomes:
\begin{equation}
     \left\|b_{k-1\sim0}\left(\tau_{k-1}, \cdots, \tau_0\right)\right\| \leq \left\|b_*\right\|.
\label{eq 2-34.1}
\end{equation}

Thus, for $\left\|\hat{{A}}_*\right\|\leq 1$, taking the limit on both sides of the equation \eqref{eq 2-32}, we have
\begin{equation}
     \begin{aligned}
     \lim _{k \rightarrow \infty}\left\|b_{k-1\sim0}\left(\tau_{k-1}, \cdots, {{\tau }}_0\right)\right\|
     & \leq \lim _{k \rightarrow \infty} \left\|b_*\right\| \frac{1-\left\|\hat{{A}}_*\right\| ^k}{1-\left\|\hat{{A}}_*\right\|} \\
     & = \frac{\left\|b_*\right\|}{1-\left\|\hat{{A}}_*\right\|}.
     \end{aligned}
     \label{eq 2-34.2}
\end{equation}

The equation \eqref{eq 2-34.1} and \eqref{eq 2-34.2} prove that under the condition \ref{condition 1}, there is an upper bound for the vector norm $\left\|b_{k- 1\sim0}\left(\tau_{k-1}, \cdots, \tau_0\right)\right\|$, and it has nothing to do with $k$.
\end{itemize}

It is easy to get from the definition of norm that when there is an element with infinite absolute value in the vector (or matrix), the norm of the vector (or matrix) is infinite. Therefore, using the nature of the inverse proposition, it can be proved that \eqref{eq 2-30} and \eqref{eq 2-32} determine that the matrix product of \eqref{eq 2-28} does not contain infinite elements . According to the definition of matrix multiplication, the result obtained by multiplying the matrix with the bounded initial disturbance vector $\varepsilon_0$ must also be bounded. Therefore, $C \in \mathbb{R}$ can always be found, satisfying
\begin{equation}
\begin{aligned}
& \lim _{k \rightarrow \infty}\left\|x_k-y_k\right\| \\
=& \lim _{k \rightarrow \infty}\left\|{A}_{k-1}\left(\tau_{k-1}\right) \cdots {A}_1\left(\tau_1\right) {A}_0\left(\tau_0\right) \varepsilon_0\right\| \\
\leq & C.
\end{aligned}
\label{eq 2-35}
\end{equation}
\end{proof}
Now, we have accomplished the proof of Proposition \eqref{proposition 2-1} that states a sufficient condition of the numerical stability, and points out a practical approach to implement validation, i.e, whether the $\left\|\hat{{A}}_k\left(\tau_k\right)\right\| \leq 1$ holds consistently. We will further conduct massive simulations to empirically support that the Condition \eqref{condition 1} holds across a wide range of driving situations.

%% file: Section06.tex
\section{Simulation}\label{sec_simulation}

In this section, a vehicle prototype is chosen from \textit{CarSim}, then a series of simulation tests are conducted to claim that: (1) our model ensures numerical stability across a wide range of situations; (2) compared to the dynamic model discretized by forward Euler method, our model remains numerically stable with a larger range of discrete time step; (3) compared to the kinematic model, our model can achieve higher accuracy; (4) our model benefits from an explicit form that makes it user-friendly to be integrated into predictive controllers for effective closed-loop control for typical stop-start task.

\subsection{Vehicle prototype}
We choose a vehicle (C-Class, Hatchback 2017) in \textit{CarSim} as the prototype. All the parameters are collected through some approximate transformations (e.g. combining the two front wheels means doubling the vertical load to the equivalent single wheel and then determine the cornering stiffness by linearizations.). 
The detailed parameters are listed in Table \ref{tab: simu paras}.

\begin{table}[tbhp]
\footnotesize
\centering
\tbl{Vehicle model parameters in simulation}
{\begin{tabular}{lll}
\toprule
Parameter & Description & Value \\
\midrule
{$I_\text{z}$} & Yaw moment of inertia &  \qty{1536.7}{kg \cdot m^2}\\
{$k_\text{f}$} & Front axle equivalent sideslip stiffness & \qty{-128916}{N/rad} \\
{$k_\text{r}$} & Rear axle equivalent sideslip stiffness & \qty{-85944}{N/rad}\\
{$l_\text{f}$} & Distance between C.G. and front axle & \qty{1.06}{m}\\
{$l_\text{r}$} & Distance between C.G. and rear axle & \qty{1.85}{m}\\
{$m$} & Mass of the vehicle & \qty{1412}{kg}\\
\bottomrule
\end{tabular}}
\label{tab: simu paras}
\end{table}

\subsection{Numerical stability verification}
For the discrete system determined by the Table \ref{tab: simu paras}, the numerical simulation can be used to verify numerical stability by judging whether the Condition \ref{condition 1} holds, i.e., $\left\|\hat{{A}}_k\left(\tau_k\right)\right\| \leq 1$. 

Specifically, according to \eqref{eq 2-25}, when the vehicle model parameters $I_z$, $k_\text{f}$, $k_\text{r}$, $l_\text{f}$, $l_\text{r}$, $m$ are known, and the discrete time step $T_\text{s}$ has been determined, then $\hat{{A}}_k$ can be determined by ${U}_{k}\left({{\tau}}_k^{(2)}\right)$ and ${U}_{ k}\left({{\tau}}_k^{(3)}\right)$. By definition, ${U}_{k}\left({{\tau}}_k^{(2)}\right)$ and ${U}_{k}\left({{\tau}}_k^{(3)}\right)$ are the weighted average of the longitudinal velocity components of $x_k$ and $y_k$ respectively, and the weighting coefficient is determined by ${{ \tau}}_k^{(2)}, {{\tau}}_k^{(3)}\in [0,1]$. Since the two state points in step $k$ may be arbitrarily chosen in the definition domain, any pairs of speeds can be substituted into $\left\|\hat{{A}}_k\left({ {\tau}}_k\right)\right\|$ to check whether the value is lower than 1 consistently.

It can be seen from the Figure \ref{A_norm} that when the longitudinal speed of the simulation is kept within $0\sim25$ \qty{}{m/s}, no matter the simulation step size $T_\text{s}$ is 0.001 s, 0.01 s or 0.1 s, it is always satisfied that  $\left\|\hat{{A}}_k\left({{\tau}}_k\right)\right\|\leq 1$. So the condition \ref{condition 1} holds consistently, i.e., the model remains numerically stable for such a wide situation range.  This means that starting from any two state initial points, as long as the control sequence is selected so that the state remains within the defined domain (here only the longitudinal vehicle speed interval limit), the state difference at the corresponding time on the two trajectories can always be guaranteed to be bounded. Intuitively, there is no divergence of simulation results due to bounded initial bias. Because if the difference between states tends to infinity, it means that at least one trajectory has a certain state component that is infinite.

\begin{figure}[htbp]
      \centering
      \captionsetup[subfigure]{justification=centering}
          \subfloat[$T_\text{s}=0.001$s]{\label{A_norm_ts0.001}\includegraphics[width=0.45\textwidth]{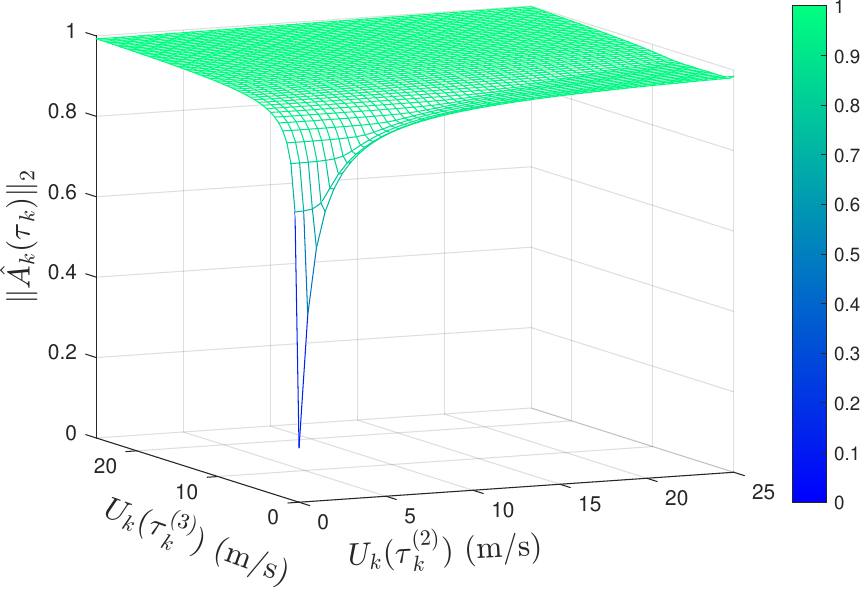}}\\
          \subfloat[$T_\text{s}=0.01$s]{\label{A_norm_ts0.01}\includegraphics[width=0.45\textwidth]{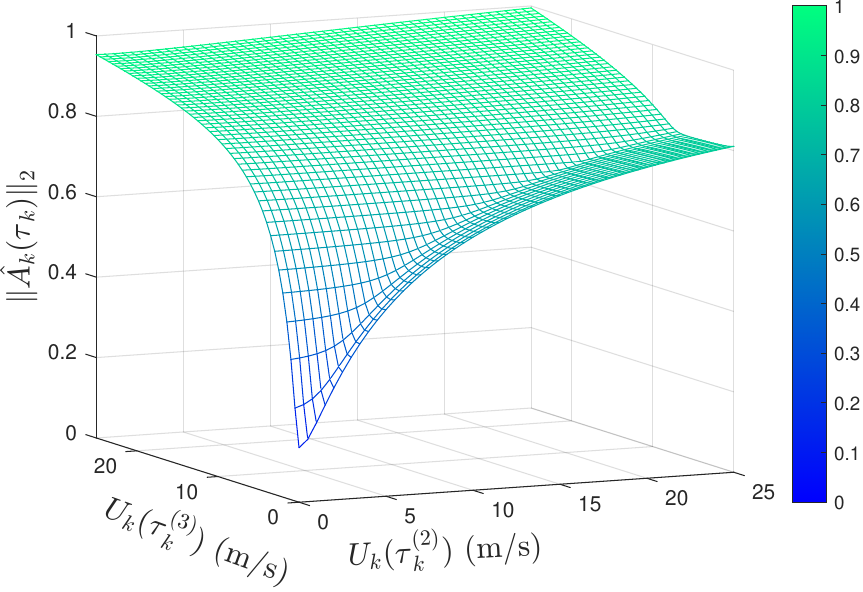}}\\
          \subfloat[$T_\text{s}=0.1$s]{\label{A_norm_ts0.1}\includegraphics[width=0.45\textwidth]{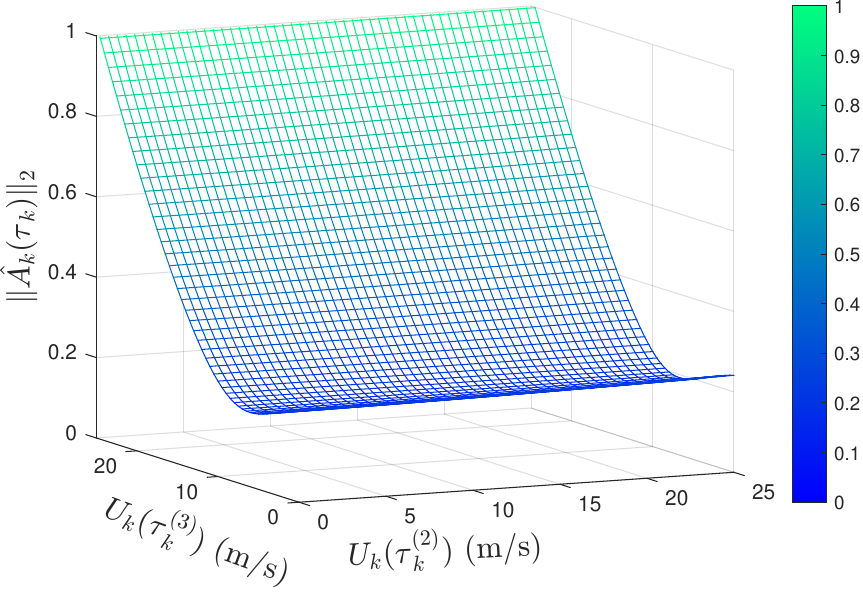}}\\
      \caption{The second norm of the sub-coefficient matrix when the longitudinal vehicle speed ranges from 0 to \qty{25}{m/s}.}
  \label{A_norm}
\end{figure}

\subsection{Time efficiency \& Robustness to noise}

In this subsection, we first conduct a computation time benchmark experiment to support the high time efficiency of our explicit dynamic vehicle model. For the experimental setup, a vehicle initialized with a velocity of $\qty{5}{m/s}$ is controlled by a step steer of $\qty{0.2674}{rad}$. We recorded the calculation times for 10,000 transition steps of both the kinematic model and our dynamic vehicle model. The results are represented in Figure \ref{time_efficiency} as boxplot visualizations. Remarkably, our model demonstrates comparable efficiency to the kinematic model, with both averaging around \qty{0.01}{ms} per transition step. This empirical evidence solidifies the high time efficiency of our explicit dynamic vehicle model.

\begin{figure}[H]
      \centering
      \captionsetup[subfigure]{justification=centering}
        {\includegraphics[width=0.52\textwidth]{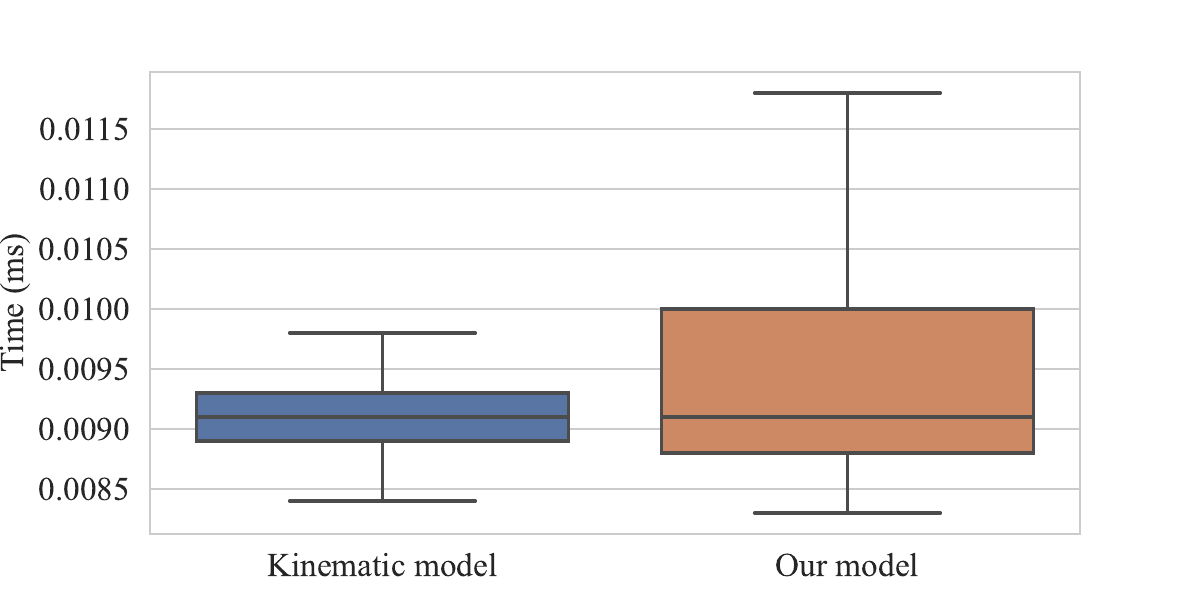}}
      \caption{Time efficiency benchmark}
  \label{time_efficiency}
\end{figure}

Secondly, robustness against disturbances or variations is also important for a realiable vehicle model. Here we conduct a robustness experiment to test whether our model can simulate realistic trajectory with the noisy input common seen in real-world. Similarly, we initialized a vehicle with a velocity of $\qty{5}{m/s}$, but introduced steering inputs of $(0.2674 + \epsilon) \qty{ }{rad}$, where $\epsilon$ follows a white noise distribution with Gaussian characteristics $\mathcal{N}(0, \sigma^2)$.
Figure \ref{robustness_test} illustrates the results for three levels of noise amplitude, i.e., $\sigma = 0.01, 0.05,  \qty{0.10}{rad}$, with the corresponding noisy steer inputs and simulated trajectories differentiated by green, blue, and orange colors, respectively. For comparison, the results of the noise-free case ($\sigma = \qty{0}{rad}$) are depicted by black dotted lines.
As the noise level increases, the deviation of the simulated trajectory grows, yet it remains remarkably close to the trajectory obtained under noise-free conditions. This observation underscores the robustness and reliability of our proposed dynamic vehicle model.

\begin{figure}[H]
      \centering
      \captionsetup[subfigure]{justification=centering}
        {\includegraphics[width=0.95\textwidth,trim={0.0cm 0.0cm 0.0cm 0.0cm}, clip]{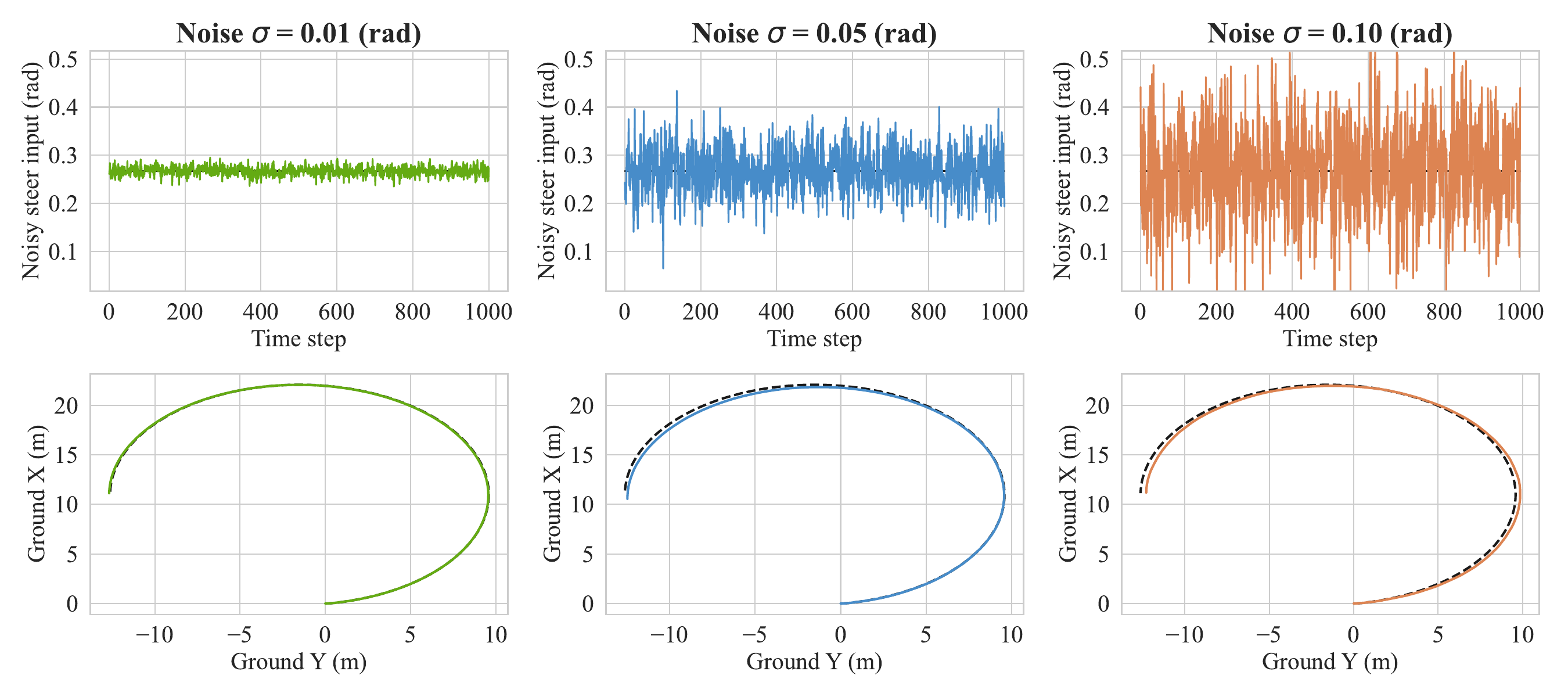}}
      \caption{Robustness test on resisting noise.}
  \label{robustness_test}
\end{figure}

\subsection{Open-loop control}

First compare the dynamic model discretized by forward Euler method. Numerical stability and insensitivity with respect to discrete time step $T_\text{s}$  and longitudinal vehicle speed $U$ are of considerable importance for the practical application of vehicle models. The open-loop control is the step steering response simulation. And the responses of our method and the dynamic model discretized by forward Euler method are calculated, at different velocities. Simultaneously, the same input is imposed on the \textit{CarSim} prototype (Adams-Moulton second-order method, $T_\text{s}$ = \qty{0.001}{s}) to serve as the groundtruth.

\begin{figure}[H]
      \centering
      \captionsetup[subfigure]{justification=centering}
          \subfloat[\ $T_\text{s}$ = 0.01 s]{\label{ts0.01}\includegraphics[width=0.9\textwidth]{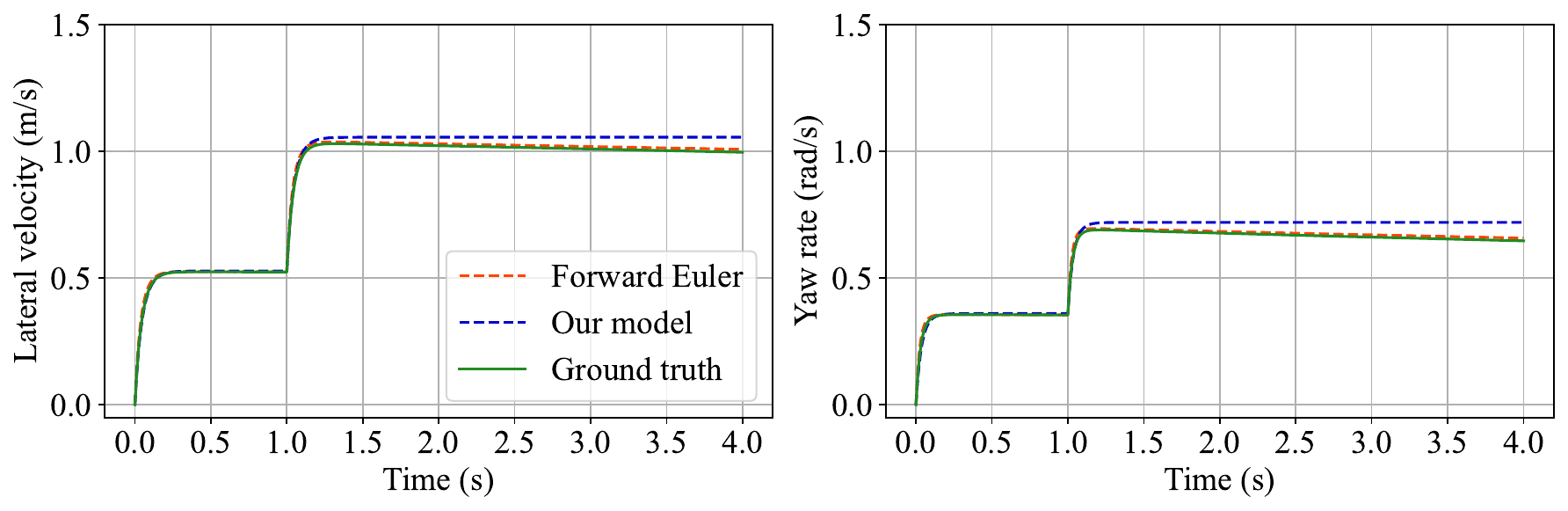}}\\
          \subfloat[$T_\text{s}$ = 0.05 s]{\label{ts0.05}\includegraphics[width=0.9\textwidth]{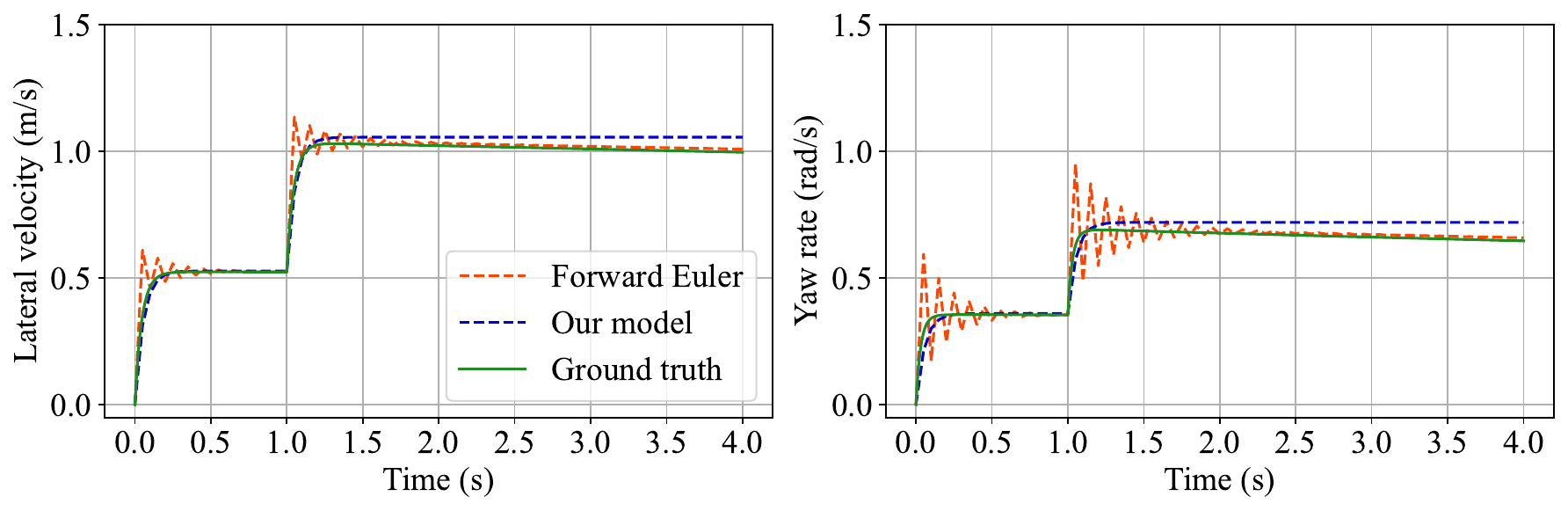}}\\
          \subfloat[$T_\text{s}$ = 0.1 s]{\label{ts0.1}\includegraphics[width=0.9\textwidth]{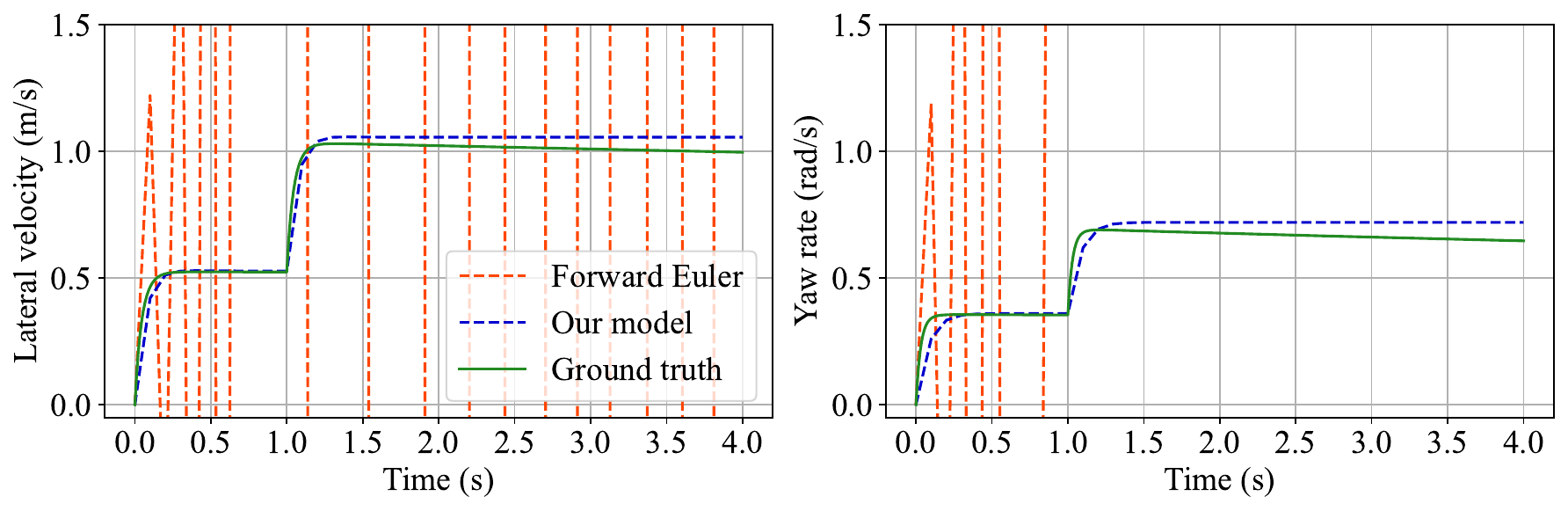}}\\
      \caption{Lateral states with varying discretization steps under two step steering input. $U_0$ = 8 m/s, the first step steering input $\delta_1$ = 0.1347 rad at 0.0 s, and the second steering input $\delta_2$ = 0.2674 rad at 1.0 s.}
  \label{back_vs_for}
\end{figure}

It can be seen from Figure \ref{back_vs_for} that our proposed model maintains numerical stability no matter the discrete time step $T_\text{s}$ is \qty{0.01}{s}, \qty{0.05}{s}, or \qty{0.1}{s}. It should be pointed out that in the process of applying the model to predictive controller design, due to the limited computing resources provided by the hardware, there is an upper bound in the prediction step of the actual nonlinear optimization problem scale, which means larger discrete time step can provide longer prediction time. The simulation results evident that our discretization method has great advantages over the forward Euler method in terms of numerical stability.

Then compare the accuracy with the kinematic model, we utilize \textit{CarSim} (employing the Adams-Moulton second-order method with a sampling time $T_\text{s}$ of 0.001s) as the reference model. Our simulations used a discrete time step of 0.001s and encompassed 25 different operational scenarios, spanning common longitudinal and lateral control ranges. Table \ref{tab: dyna_vs_kine} presents the specific errors relative to \textit{CarSim}. Notably, there is a considerable reduction in errors compared to the kinematic model, often exceeding $90\%$. This suggests that our vehicle model holds potential for enhancing the control accuracy of model-based controllers.
\begin{table}[htbp]
    \footnotesize
    \caption{RMS Error of trajectory under varying velocities and steer inputs.}
    \label{tab: dyna_vs_kine}
    \begin{center}
    \begin{tabular}{ccccc} 
    \toprule
    $U_0$ (\qty{}{m/s})&$\delta$ (\qty{}{rad})&Kinematic (\qty{}{m}) &Our model (\qty{}{m})&Accuracy improvement\\
    \midrule

   5 &  0.05 & 0.056& 0.015& +$74.31\%$\\
   5 &  0.10 & 0.119& 0.029& +$76.08\%$\\
   5 &  0.15 & 0.195& 0.042& +$78.59\%$\\
   5 &  0.20 & 0.290& 0.054& +$81.42\%$\\
   5 &  0.25 & 0.408& 0.064& +$84.24\%$\\
   \midrule
   10 &  0.05 & 0.488& 0.050& +$89.80\%$\\
   10 &  0.10 & 0.972& 0.095& +$90.22\%$\\
   10 &  0.15 & 1.449& 0.132& +$90.88\%$\\
   10 &  0.20 & 1.907& 0.158& +$91.71\%$\\
   10 &  0.25 & 2.332& 0.171& +$92.66\%$\\
   \midrule
   15 &  0.05 & 1.736& 0.096& +$94.46\%$\\
   15 &  0.10 & 3.315& 0.177& +$94.67\%$\\
   15 &  0.15 & 4.591& 0.229& +$95.02\%$\\
   15 &  0.20 & 5.451& 0.247& +$95.46\%$\\
   15 &  0.25 & 5.843& 0.235& +$95.98\%$\\
   \midrule
   20 &  0.05 & 4.229& 0.145& +$96.58\%$\\
   20 &  0.10 & 7.710& 0.253& +$96.71\%$\\
   20 &  0.15 & 9.868& 0.303& +$96.93\%$\\
   20 &  0.20 & 10.483& 0.293& +$97.21\%$\\
   20 &  0.25 & 9.941& 0.253& +$97.45\%$\\
   \midrule
   25 &  0.05 & 8.305& 0.188& +$97.73\%$\\
   25 &  0.10 & 14.391& 0.313& +$97.82\%$\\
   25 &  0.15 & 16.895& 0.345& +$97.96\%$\\
   25 &  0.20 & 16.065& 0.308& +$98.08\%$\\
   25 &  0.25 & 14.620& 0.282& +$98.07\%$\\
    \bottomrule
    \end{tabular}
    \end{center}
\end{table}

\subsection{Closed-loop control}
The closed-loop control is the design of a nonlinear model predictive controller (MPC) to perform the starting task. We set discrete time step $T_\text{s}$ = \qty{0.1}{s}, prediction step $N_\text{p}=20$, and control step $N_\text{c}=1$. In such a predictive controller, the receding-horizon self-correction will help compensate for the model error dynamically. We design a nonlinear model predictive controller as  (\ref{back_euler_bicycle}) to comprehensively verify its effect in practical application. Starting from $(0, 0)$, the vehicle model first tracks a direct path connecting its center of gravity (C.G.) and the target $(30, 30)$. The reference speed is set to be $\qty{6}{m/s}$. There is an obstacle (light orange round area in Figure \ref{closed_loop_traj}) located at $(15, 15)$ blocking its way. The parameters are deliberately co-designed so that the vehicle can not bypass with steering, but only stop to avoid collision. Note that collision is defined by overlapping of the obstacle and C.G. of the vehicle, instead of vehicle body profile. Once the vehicle stops, the obstacle is moved to $(18, 12)$ (dark orange round area in Figure \ref{closed_loop_traj}), giving way to the vehicle. Then the vehicle starts to execute a stop-start task. To be emphasized, obstacle avoidance requires the vehicle to steer and accelerate simultaneously, which is similar to a left-turning maneuver after red light at intersection.

The constrained nonlinear optimal control problem is formulated as (\ref{ocp}). The longitudinal velocity is shown in Figure \ref{velocity}, whereas the longitudinal input $a$ and lateral input $\delta$ are recorded in Figure \ref{acc} and \ref{steer} respectively. 

\begin{subequations}
    \begin{align}
        \min \sum_{k=0}^{N_p-1}\left(x_k-x_{\text{ref}, k}\right)^{\top} & Q\left(x_k-x_{\text{ref}, k}\right)+u_k^{\top} R u_k ,\\
        x_{k+1}&=f(x_k, u_k),\\
        \left(x_k-x_{\text{obs}, k}\right)^{\top} &Q_{\text{s}}\left(x_k-x_{\text{obs}, k}\right) \geq D_{\text{s}}^{2} ,\\
        x_{\text{min} } \leq x_k \leq x_{\text{max} } &,
        u_{\text{min} } \leq u_k \leq u_{\text{max} },
    \end{align}
    \label{ocp}
\end{subequations}
where $x_{\text{ref}, k}$ is the $k$th reference point in the prediction horizon. $x_{\text{obs}, k}$ conveys location of the obstacle at $k$th step in the prediction horizon. $Q$ = diag(100, 100, 0, 0, 0, 0), $R$ = diag(10, 500), $Q_{\text{s}}$ = diag(1, 1, 0, 0, 0, 0), $D_{\text{s}}$ = 8. $x_{\text{min}}$ = [-$\infty$, -$\infty$, -$\infty$, 0, -4, -3], $x_{\text{max}}$ = [+$\infty$, +$\infty$, +$\infty$, 20, 4, 3]. $u_{\text{min}}$ = [-5, -$\pi$/4], $u_{\text{max}}$ = [2, $\pi$/4]. All values are in the international system of units. 

\begin{figure}[htbp]
    \centering
    \captionsetup[subfigure]{justification=centering}
        \subfloat[Stop]{\label{stop}\includegraphics[width=0.4\textwidth]{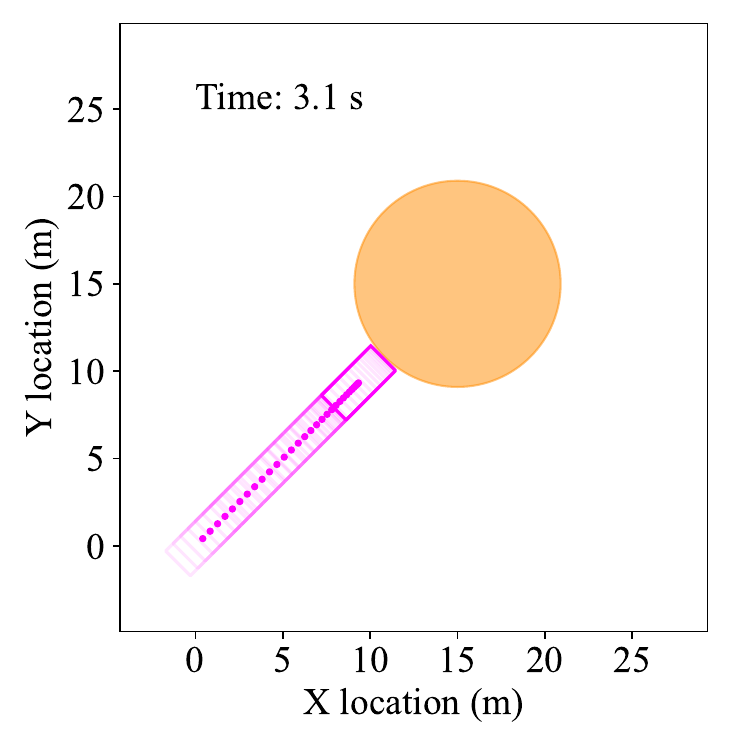}} \quad
        \subfloat[Start]{\label{go}\includegraphics[width=0.4\textwidth]{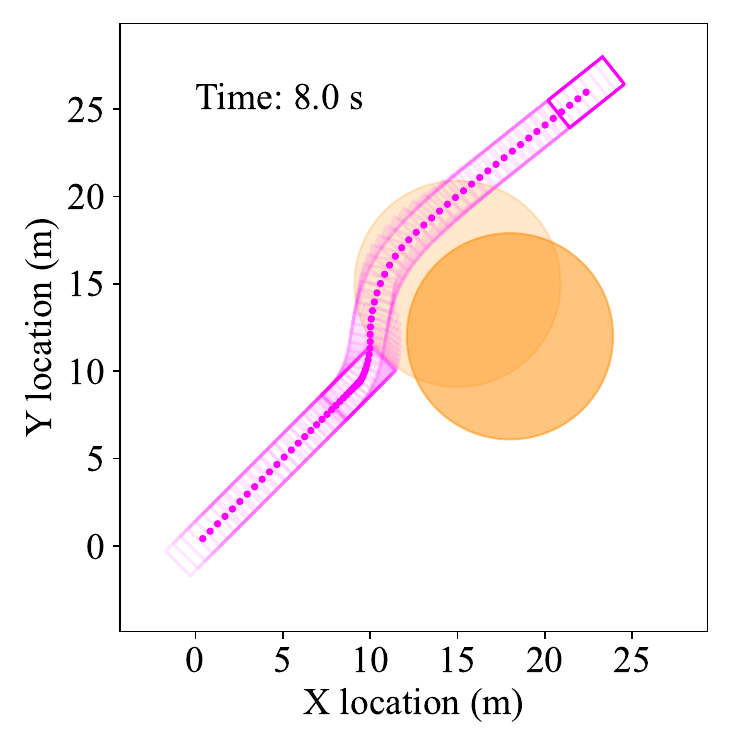}}
    \caption{Trajectory of the close-loop control process of a stop-start driving scenario. The obstacle changes position from $(15, 15)$ to $(18, 12)$ when the time is \qty{3.1}{s}.}
    \label{closed_loop_traj}
\end{figure}

\begin{figure}[H]
    \centering
    \captionsetup[subfigure]{justification=centering}
        \subfloat[Longitudinal velocity]{\label{velocity}\includegraphics[width=0.32\textwidth]{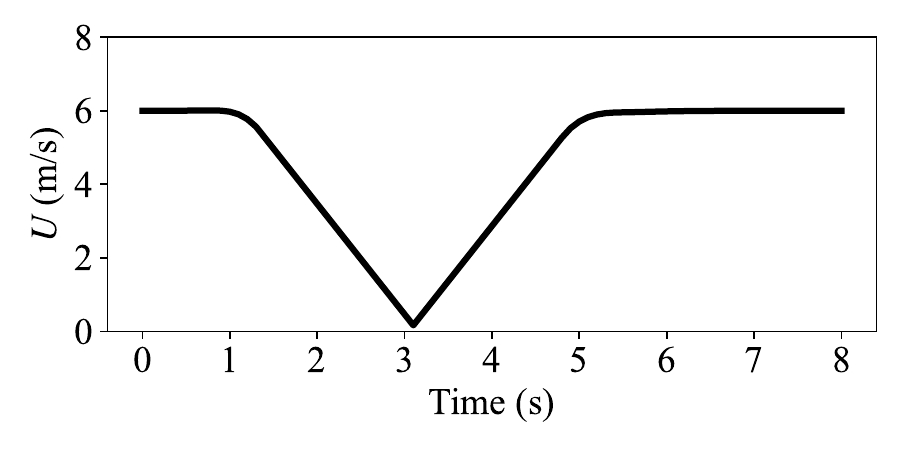}}
        \subfloat[Longitudinal action] {\label{acc}\includegraphics[width=0.32\textwidth]{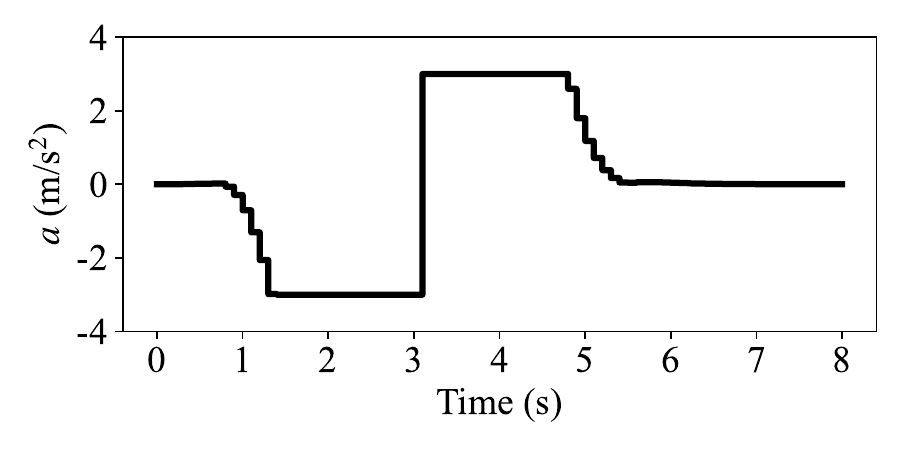}}
        \subfloat[Lateral action]{\label{steer}\includegraphics[width=0.32\textwidth]{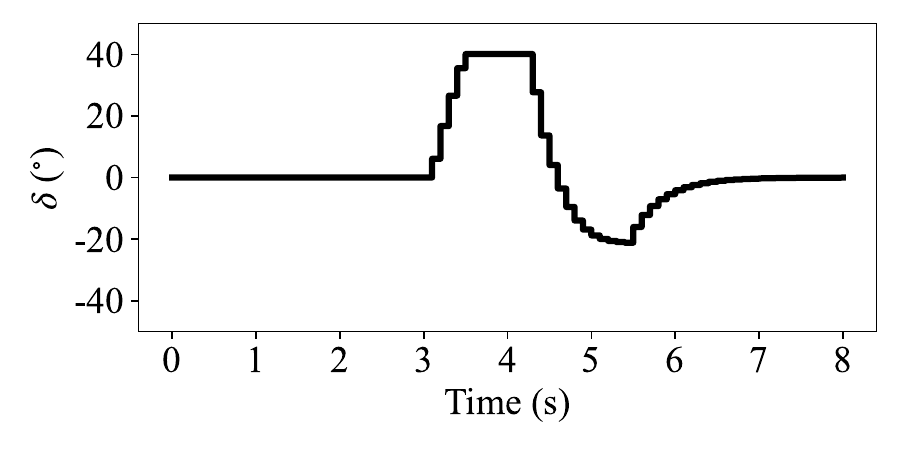}}\\
    \caption{Control input of the closed-loop control experiment.}
\label{closed_loop}
\end{figure}

%% file: Section07.tex
\section{Real Vehicle
Experiments}\label{sec_experiment}
This section will further validate the performance of our proposed dynamic model through real vehicle experiments. The equipments, settings and results will be presented as follows.
\subsection{Experiment equipments}
In the experiment, the vehicle we adopt is Changan CS55 E-SUV. It is equipped with an ADLINK ROS-CUBE-I Industrial Personal Computer (IPC). 
What's more, the vehicle is equipped with perception and localization systems, which receives and records information at a frequency of \qty{10}{Hz}. For these sensors, the position accuracy is \qty{0.01}{m}, yaw angle accuracy is \qty{0.001}{rad}. The snapshots of the vehicle and the IPC are shown in Figure \ref{fig:real veh exp hardward}. Before the experiments, we identified the vehicle's parameters including $k_\text{f}$, $k_\text{r}$, $l_\text{f}$, $l_\text{r}$, $m$, and $I_\text{z}$ through preliminary trials and measurements. The specific numerical values are presented in Table \ref{tab: real paras}.

\begin{figure}[htbp]
    \centering
    \captionsetup[subfigure]{justification=centering}
        \subfloat[Experimantal vehicle]{\label{Experimantal vehicle}\includegraphics[width=0.4\textwidth]{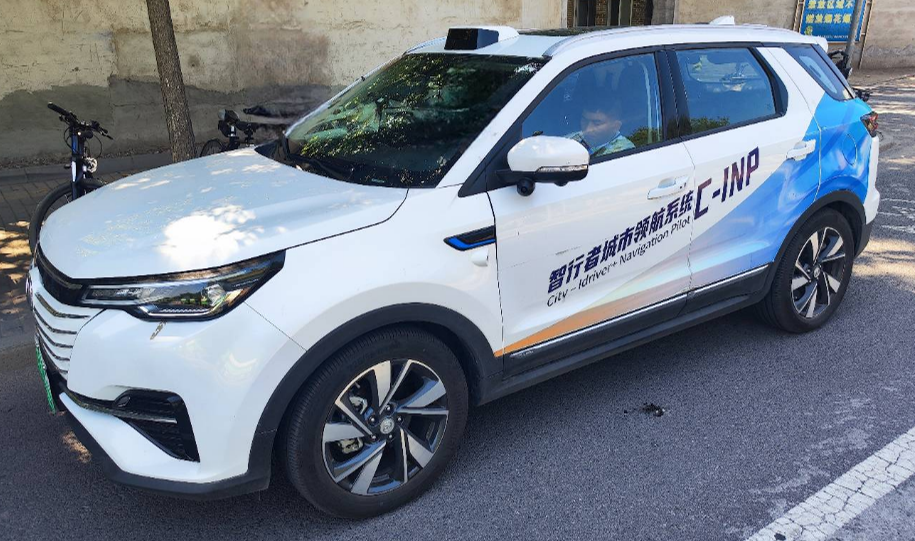}}
        \subfloat[Industrial personal computer]{\label{Industrial personal computer}\includegraphics[width=0.4\textwidth, trim={0cm, 0.8cm, 0cm, 0cm}, clip]{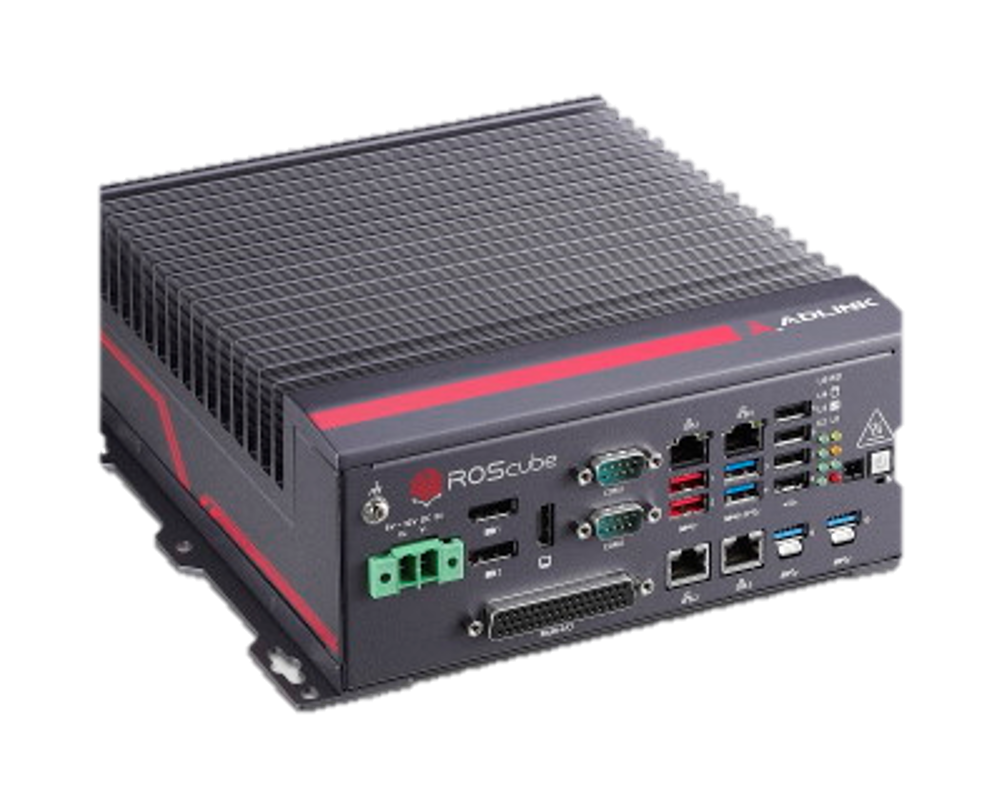}}\\
    \caption{Hardwares in real vehicle experiments}
    \label{fig:real veh exp hardward}
\end{figure}

\begin{table}[tbhp]
\footnotesize
\centering
\tbl{Vehicle model parameters in real vehicle experiment.}
{\begin{tabular}{lll}
\toprule
Parameter & Description & Value \\
\midrule
{$I_{z}$} & Yaw moment of inertia &  \qty{3058}{kg \cdot m^2}\\
{$k_\text{f}$} & Front axle equivalent sideslip stiffness & \qty{-186000}{N/rad} \\
{$k_\text{r}$} & Rear axle equivalent sideslip stiffness & \qty{-183000}{N/rad}\\
{$l_\text{f}$} & Distance between C.G. and front axle & \qty{1.4}{m}\\
{$l_\text{r}$} & Distance between C.G. and rear axle & \qty{1.5}{m}\\
{$m$} & Mass of the vehicle & \qty{1892}{kg}\\
$T_\text{s}$ & Discrete time step & \qty{0.1}{s}\\
\bottomrule
\end{tabular}}
\label{tab: real paras}
\end{table}

\subsection{Experiment settings}
Our vehicle experiments are conducted in open-loop control manner, issuing motion commands to the chassis through the IPC. We firstly validate the proposed dynamic model from the longitudinal aspect. Building upon this foundation, we then proceed to validate the accuracy of the dynamic model from the combined longitudinal and lateral aspects. In the longitudinal experimental series, we maintain a zero steering angle to ensure straight-line driving of the vehicle. Afterwards, the vehicle accelerates from zero to its maximum velocity, followed by deceleration back to zero. Based on variations in acceleration, deceleration, and maximum velocity, a total of six experiments were conducted as outlined in Table \ref{tab: real_exp_setting lon}. In the lateral experimental series, we orchestrate the vehicle's acceleration from zero to its maximum velocity while applying a constant steering angle throughout the acceleration phase. Based on variations in acceleration, steering angle, and maximum velocity, a total of eight experiments were conducted as delineated in Table \ref{tab: real_exp_setting lat}.

\subsection{Experiment results}
For a fair comparison between real vehicle and our proposed vehicle model, we utilize the actually measured values of acceleration, and front-wheel steering angle obtained from sensors to implement open-loop control simulation. The discrete time step of vehicle simulation is \qty{0.1}{s}, which matches the sampling frequency of the real vehicle state information. To show the effectiveness of fitting, we calculated the RMS error of the real trajectory and the simulated trajectory for each experiment. For the longitudinal experimental series, these results are presented in Table \ref{tab: real_exp_setting lon}, while the results of the lateral experiments are presented in Table \ref{tab: real_exp_setting lat}. Moreover, to give a better understanding of the effectiveness, we provide visualizations for two typical cases from the longitudinal and lateral results, respectively.

\begin{table}[htbp]
    \footnotesize
    \centering
    \caption{Longitudinal real vehicle experiments}
    \label{tab: real_exp_setting lon}
    \begin{tabular}{ccccc}
    \toprule
    \textbf{Id} &\textbf{Acceleration ($\qty{}{m/s^2}$)}  &\textbf{Deceleration ($\qty{}{m/s^2}$)} &\textbf{Maximum velocity ($\qty{}{m/s}$)}& \textbf{RMS error ($\qty{}{m}$)}\\
    \midrule
    Lon-1 & 1.0 & -2.0 & 6  & 0.641\\
    Lon-2 & 1.5 & -1.0 & 6  & 0.362\\
    Lon-3 & 0.5 & -3.0 & 10 & 1.358\\
    Lon-4 & 1.0 & -3.0 & 10 & 1.740\\
    Lon-5 & 2.0 & -3.0 & 10 & 1.329\\
    Lon-6 & 1.5 & -3.0 & 10 & 1.707\\
    \bottomrule
    \end{tabular}
\end{table}

\begin{table}[htbp]
    \footnotesize
    \centering
    \caption{Lateral real vehicle experiments}
    \label{tab: real_exp_setting lat}
    \begin{tabular}{ccccc}
    \toprule
    \textbf{Id} &\textbf{Acceleration ($\qty{}{m/s^2}$)} &\textbf{Maximum velocity ($\qty{}{m/s}$)}&\textbf{Steering angle ($\qty{}{rad}$)}& \textbf{RMS error ($\qty{}{m}$)}\\
    \midrule
    Lat-1 & 1.5 & 3  & 0.0210 & 0.561\\
    Lat-2 & 1.5 & 3  & 0.0526 & 0.938\\
    Lat-3 & 1.5 & 3  & 0.1051 & 0.753\\
    Lat-4 & 1.5 & 3  & 0.2103 & 0.687\\
    Lat-5 & 1.5 & 6  & 0.0210 & 1.499\\
    Lat-6 & 1.5 & 6  & 0.0526 & 1.267\\
    Lat-7 & 1.5 & 6  & 0.1051 & 1.105\\
    Lat-8 & 1.5 & 10 & 0.0526 & 2.618\\
    \bottomrule
    \end{tabular}
\end{table}

 As shown in Figure \ref{fig: real_exp_lon} and \ref{fig: real_exp_lat}, the real vehicle longitudinal velocity and the longitudinal velocity simulated by the dynamic model are plotted together on one subfigure, while the real vehicle trajectory and the velocity simulated by the our model are illustrated on another subFigure Based on the results, it is evident that the proposed dynamic model in this paper accurately captures the vehicle's motion across all different cases. 
 This alignment between the model and real vehicle motion reinforces the remarkable superiority of our proposed dynamic model, especially evident in low-speed stop-start scenarios. While there is a trend of increased RMS errors between the real and simulated trajectories as the maximum vehicle velocity and steering angle rise, the RMS errors remain within an acceptable range.

\begin{figure}[H]
      \centering
      \captionsetup[subfigure]{justification=centering}
          \subfloat[Lon-1 longitudinal velocity]{\label{Lon-1 velocity}\includegraphics[width=0.5\textwidth]{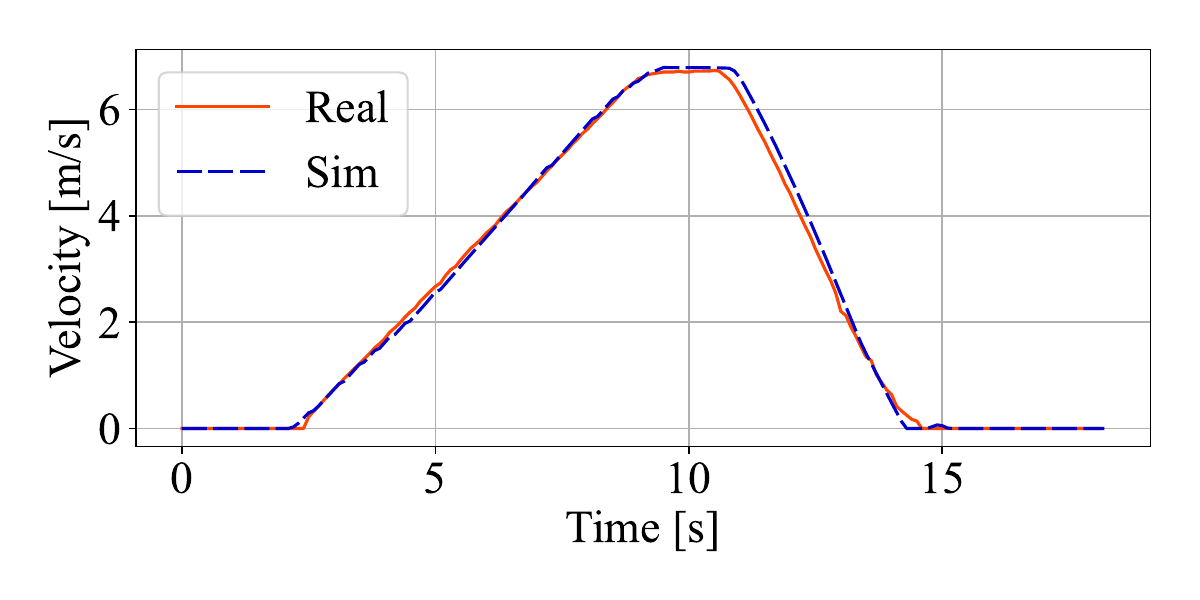}}
          \subfloat[Lon-1 trajectory]{\label{Lon-1 trajectory}\includegraphics[width=0.5\textwidth]{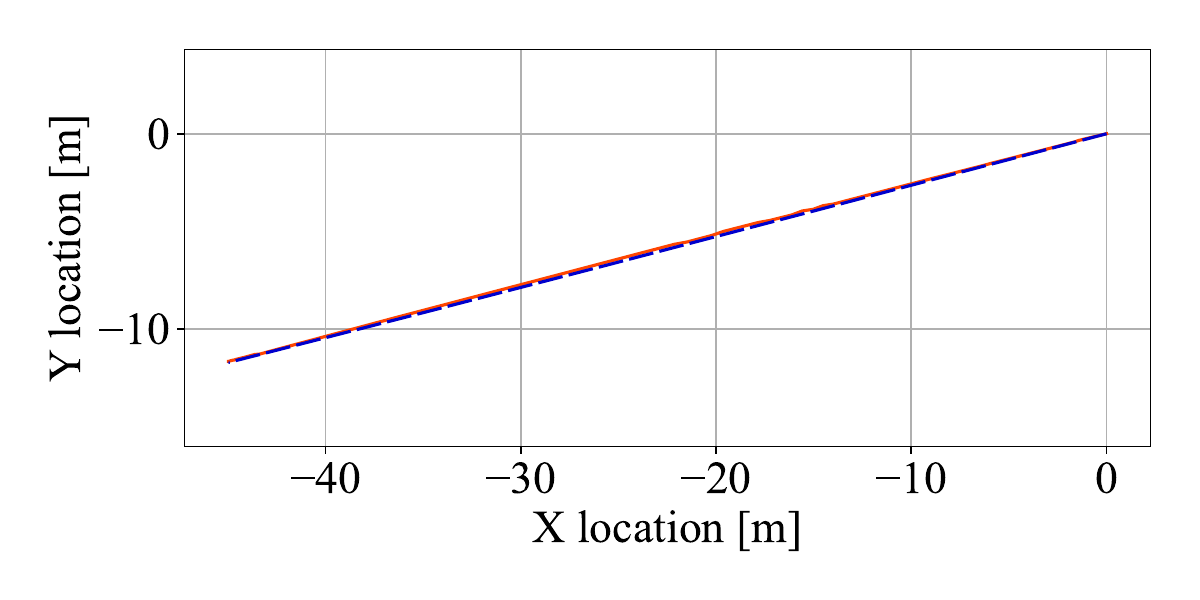}}\\
          \subfloat[Lon-5 longitudinal velocity]{\label{Lon-2 velocity}\includegraphics[width=0.5\textwidth]{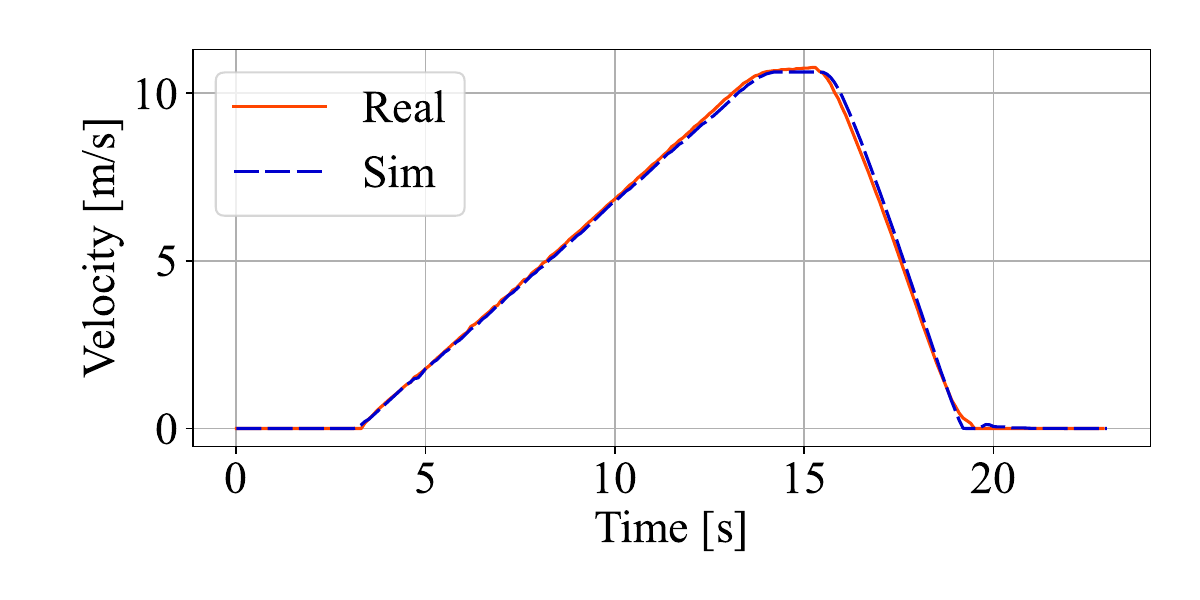}}
          \subfloat[Lon-5 trajectory]{\label{Lon-2 trajectory}\includegraphics[width=0.5\textwidth]{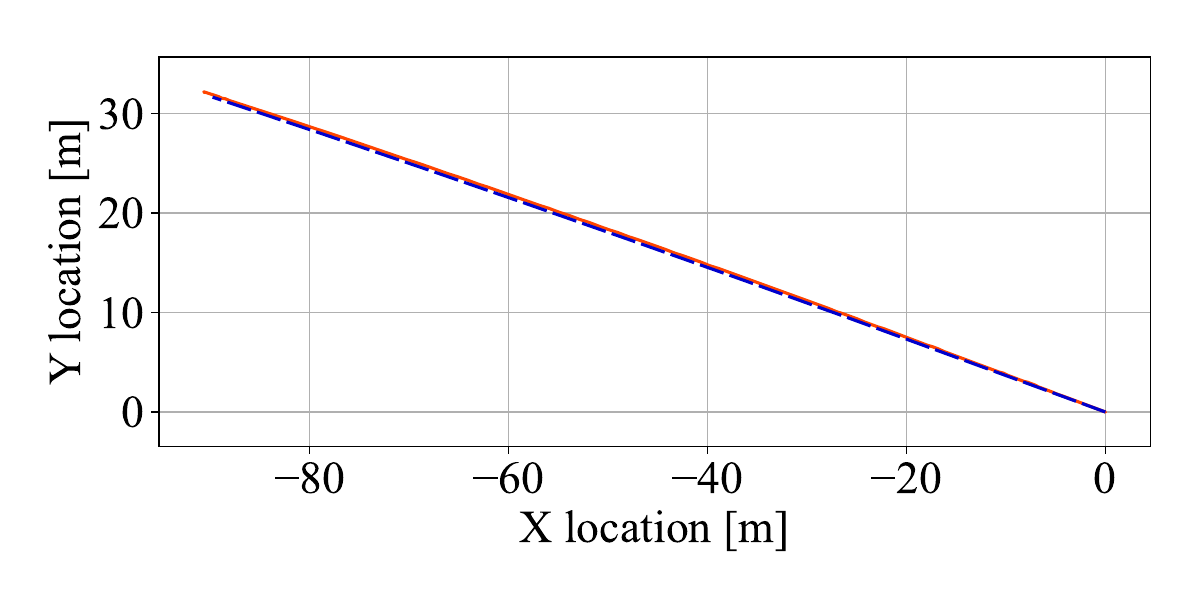}}\\
      \caption{{Longitudinal real vehicle experiments}}
  \label{fig: real_exp_lon}
\end{figure}

\begin{figure}[H]
      \centering
      \captionsetup[subfigure]{justification=centering}
          \subfloat[Lat-1 longitudinal velocity]{\label{Lat-2 velocity}\includegraphics[width=0.5\textwidth]{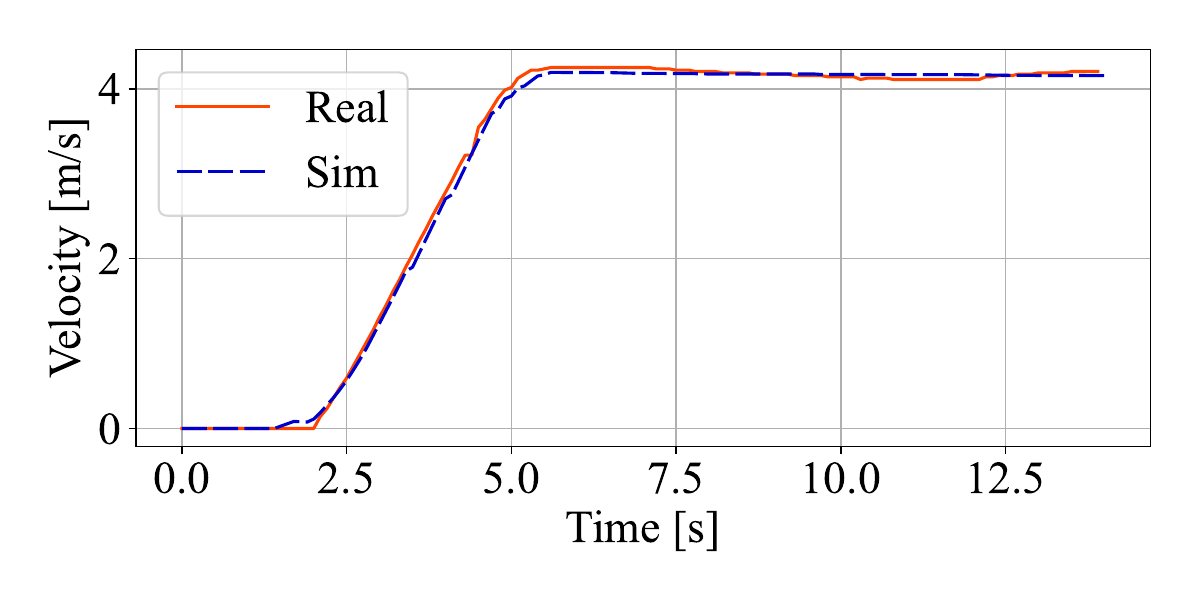}}
          \subfloat[Lat-1 trajectory]{\label{Lat-2 trajectory}\includegraphics[width=0.5\textwidth]{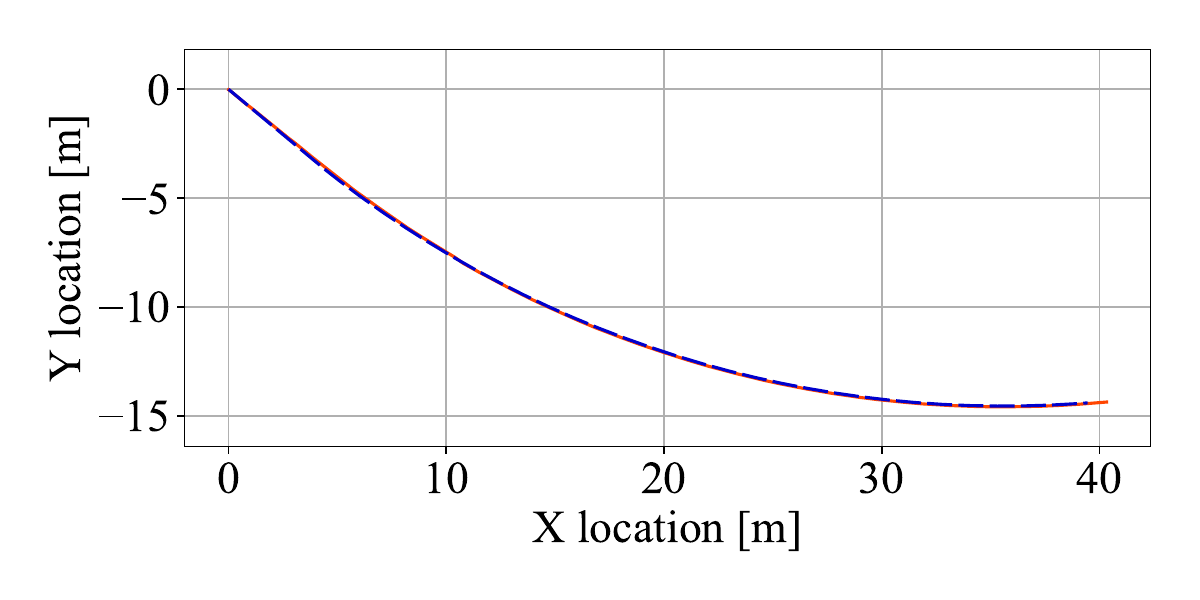}}\\
          \subfloat[Lat-2 longitudinal velocity]{\label{Lat-12 velocity}\includegraphics[width=0.5\textwidth]{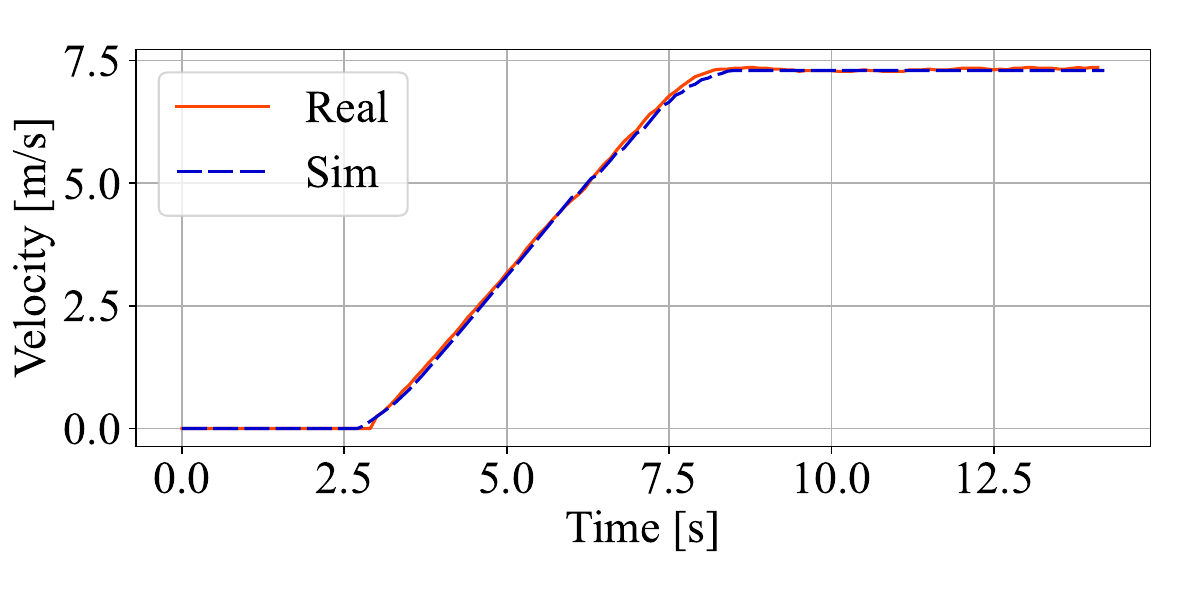}}
          \subfloat[Lat-2 trajectory]{\label{Lat-12 trajectory}\includegraphics[width=0.5\textwidth]{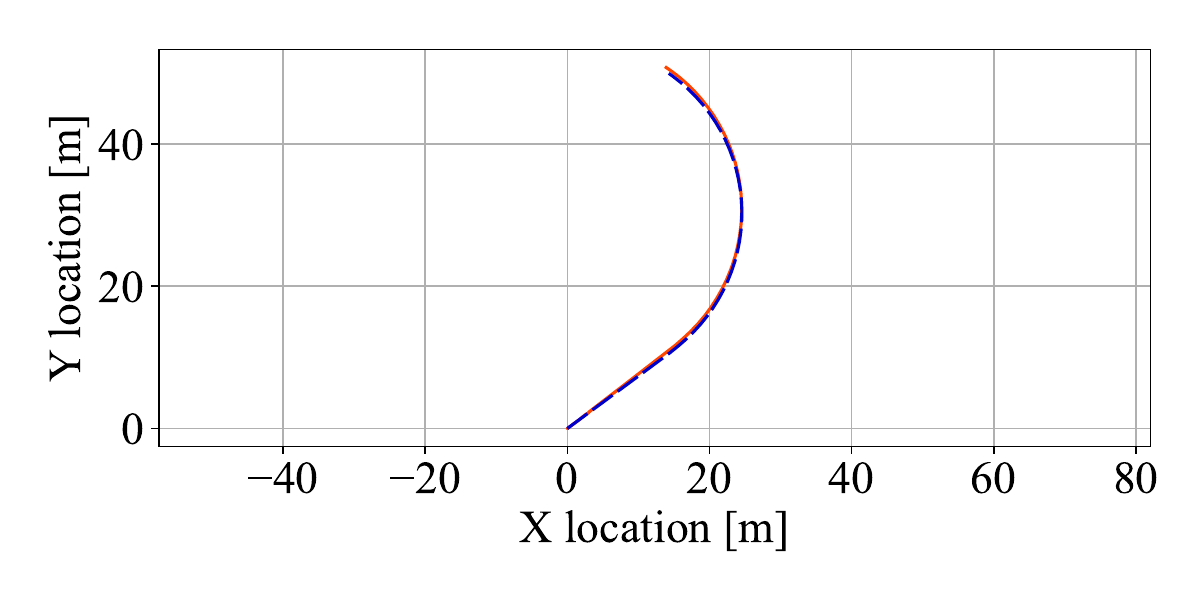}}\\
      \caption{{Lateral real vehicle experiments}}
  \label{fig: real_exp_lat}
\end{figure}

%% file: Section08.tex
\section{Conclusion}\label{sec_conclusion}

In this paper, we propose an explicit discretized dynamic vehicle model and provide a sufficient condition that guarantees its numerical stability. Our proposed model outperforms forward-Euler discretized dynamic model in terms of stability and the kinematic model in terms of accuracy. Furthermore, the real vehicle experiments are further conducted to demonstrate that our proposed model can closely aligns to the real vehicle trajectories showcasing its practicality and ease of use. In conclusion, these simulations and experiments empirically support that our proposed explicit dynamic vehicle model is user-friendly and practically suitable for predictive controller design, model-based learning algorithm, etc., under general driving scenarios including
low-speed and stop-start.

%% file: main.bbl
\begin{thebibliography}{10}

\bibitem{sharp2011vehicle}
RS~Sharp and Huei Peng.
\newblock Vehicle dynamics applications of optimal control theory.
\newblock {\em Vehicle System Dynamics}, 49(7):1073--1111, 2011.

\bibitem{zhan2023continuous}
Guojian Zhan, Yuxuan Jiang, Jingliang Duan, Shengbo~Eben Li, Bo~Cheng, and Keqiang Li.
\newblock Continuous-time policy optimization.
\newblock In {\em 2023 American Control Conference (ACC)}, pages 3382--3388. IEEE, 2023.

\bibitem{yang2013overview}
Shaopu Yang, Yongjie Lu, and Shaohua Li.
\newblock An overview on vehicle dynamics.
\newblock {\em International Journal of Dynamics and Control}, 1:385--395, 2013.

\bibitem{kutluay2014validation}
Emir Kutluay and Hermann Winner.
\newblock Validation of vehicle dynamics simulation models--a review.
\newblock {\em Vehicle System Dynamics}, 52(2):186--200, 2014.

\bibitem{manrique2022analytical}
Camilo~Andr{\'e}s Manrique-Escobar, Carmine~Maria Pappalardo, and Domenico Guida.
\newblock On the analytical and computational methodologies for modelling two-wheeled vehicles within the multibody dynamics framework: a systematic literature review.
\newblock {\em Journal of Applied and Computational Mechanics}, 8(1):153--181, 2022.

\bibitem{kong2015kinematic}
Jason Kong, Mark Pfeiffer, Georg Schildbach, and Francesco Borrelli.
\newblock Kinematic and dynamic vehicle models for autonomous driving control design.
\newblock In {\em 2015 IEEE Intelligent Vehicles Symposium (IV)}, pages 1094--1099, Seoul, Korea, 2015. IEEE.

\bibitem{arnold2011numerical}
Martin Arnold, Bernhard Burgermeister, Claus F{\"u}hrer, Gerhard Hippmann, and Georg Rill.
\newblock Numerical methods in vehicle system dynamics: state of the art and current developments.
\newblock {\em Vehicle System Dynamics}, 49(7):1159--1207, 2011.

\bibitem{biswas2013discussion}
BN~Biswas, Somnath Chatterjee, SP~Mukherjee, and Subhradeep Pal.
\newblock A discussion on euler method: A review.
\newblock {\em Electronic Journal of Mathematical Analysis and Applications}, 1(2):2090--2792, 2013.

\bibitem{lugner2005recent}
Peter Lugner, Hans Pacejka, and Manfred Pl{\"o}chl.
\newblock Recent advances in tyre models and testing procedures.
\newblock {\em Vehicle System Dynamics}, 43(6-7):413--426, 2005.

\bibitem{10056957}
Nir Shlezinger, Jay Whang, Yonina~C. Eldar, and Alexandros~G. Dimakis.
\newblock Model-based deep learning.
\newblock {\em Proceedings of the IEEE}, 111(5):465--499, 2023.

\bibitem{li2023brlok}
Shengbo~Eben Li.
\newblock {\em Reinforcement learning for sequential decision-making and control}.
\newblock Springer, 2023.

\bibitem{li2010model}
Shengbo~Eben Li, Keqiang Li, Rajesh Rajamani, and Jianqiang Wang.
\newblock Model predictive multi-objective vehicular adaptive cruise control.
\newblock {\em IEEE Transactions on control systems technology}, 19(3):556--566, 2010.

\bibitem{ren2023improve}
Yangang Ren, Guojian Zhan, Liye Tang, Shengbo~Eben Li, Jianhua Jiang, Keqiang Li, and Jingliang Duan.
\newblock Improve generalization of driving policy at signalized intersections with adversarial learning.
\newblock {\em Transportation Research Part C: Emerging Technologies}, 152:104161, 2023.

\bibitem{polack2017kinematic}
Philip Polack, Florent Altch{\'e}, Brigitte d'Andr{\'e}a Novel, and Arnaud de~La~Fortelle.
\newblock The kinematic bicycle model: A consistent model for planning feasible trajectories for autonomous vehicles?
\newblock In {\em 2017 IEEE Intelligent Vehicles Symposium (IV)}, pages 812--818, Los Angeles, CA, USA, 2017. IEEE.

\bibitem{pacejka1992magic}
Hans~B Pacejka and Egbert Bakker.
\newblock The magic formula tyre model.
\newblock {\em Vehicle System Dynamics}, 21(S1):1--18, 1992.

\bibitem{hirschberg2007tire}
Wolfgang Hirschberg, Georg Rill, and Heinz Weinfurter.
\newblock Tire model tmeasy.
\newblock {\em Vehicle System Dynamics}, 45(S1):101--119, 2007.

\bibitem{bernard1995tire}
Jim~E Bernard and Chris~L Clover.
\newblock Tire modeling for low-speed and high-speed calculations.
\newblock {\em SAE Transactions}, pages 474--483, 1995.

\end{thebibliography}
